\documentclass[reqno,oneside,11pt]{amsart}
\usepackage[T1]{fontenc}
\usepackage[utf8]{inputenc}
\usepackage[top=1.5cm,bottom=2cm,right=2cm,left=2cm]{geometry}
\usepackage{amsmath,amsfonts,amssymb,amscd,amsthm}
\usepackage[usenames,dvipsnames,svgnames]{xcolor}
\usepackage[numbers]{natbib}
\usepackage{doi}
\usepackage{graphicx}
\usepackage{caption}
\usepackage{subcaption}
\usepackage{autonum} % Only enumerate what is referred
\usepackage{hyperref} % hyperref insertions

\usepackage{acronym}
 
\usepackage{tikz}
\definecolor{myellow}{RGB}{255,230,128}
\definecolor{gray20}{RGB}{204,204,204}

\acrodef{fe}[FE]{finite element}
\acrodef{dof}[DOF]{degree of freedom}
\acrodefplural{dof}[DOFs]{degrees of freedom}
\acrodef{vef}[VEF]{vertex, edge, or face}
\acrodefplural{vef}[VEFs]{vertices, edges, or faces}
\acrodef{dg}[DG]{discontinuous Galerkin}
 
\newtheorem{theorem}{Theorem}[section]
\newtheorem{lemma}[theorem]{Lemma}

\newtheorem{corollary}[theorem]{Corollary}

\newtheorem{remark}[theorem]{Remark}

\newtheorem{method}[theorem]{Algorithm}

\makeatletter
\newcommand{\opnorm}{\@ifstar\@opnorms\@opnorm}
\newcommand{\@opnorms}[1]{%
  \left|\mkern-1.5mu\left|\mkern-1.5mu\left|
   #1
  \right|\mkern-1.5mu\right|\mkern-1.5mu\right|
}
\newcommand{\@opnorm}[2][]{%
  \mathopen{#1|\mkern-1.5mu#1|\mkern-1.5mu#1|}
  #2
  \mathclose{#1|\mkern-1.5mu#1|\mkern-1.5mu#1|}
}
\makeatother

\def\opnormh#1{\opnorm{#1}_h}
\def\Eq#1{(\ref{eq:#1})}

\def\x{\boldsymbol{x}}
\def\Dom{\Omega}
\def\Domap{{\FV\Dom}}
\def\Domex{{\Omega_{\rm act}}}
\def\Domin{{\Omega_{\rm in}}}
\def\Domart{{\Omega_{\rm art}}}
\def\Domext{{\Omega_{\rm ext}}}
\def\Domcut{{\Omega_{\rm cut}}}

\def\meshex{\mathcal{T}^{\rm act}_h}
\def\meshin{\mathcal{T}^{\rm in}_h}

\def\meshact{\mathcal{T}^{\rm act}_h}
\def\meshart{\mathcal{T}^{\rm art}_h}
\def\meshext{\mathcal{T}^{\rm ext}_h}
\def\meshin{\mathcal{T}^{\rm in}_h}
\def\meshcut{\mathcal{T}^{\rm cut}_h}
\def\meshag{\mathcal{T}^{\rm agg}_h}

\def\cev{{C_{\rm e}}}
\def\cev{{C_{\rm e}}}

\def\cbou{{C_{\partial}}}

\def\cell{K}
\def\aggr{{A_K}}
 
\def\cellref{{\hat K}}
\def\aggrref{\hat{A}_\cell}

\def\uh{u_h}
\def\vh{v_h}
\def\shpf#1{\phi^{#1}}

\def\lset{\psi}

\def\hc{h_\cell}

\def\h{h}

\def\fespst{{V_h}}
\def\fespin{{V_h^{\rm in}}}
\def\fespex{{V_h^{\rm act}}}
\def\fespag{{V_h^{\rm agg}}}
\def\fespagnew{{V_h^{\rm agg,*}}}

\def\nodes#1{\mathcal{N}(#1)}

\def\nodesgin{\mathcal{N}_h^{\rm in}}
\def\nodesgex{\mathcal{N}_h^{\rm act}}
\def\nodesgou{\mathcal{N}_h^{\rm out}}

\def\uvin{{\underline{\mathbf{u}}^{\rm in}}}
\def\uvex{{\underline{\mathbf{u}}^{\rm act}}}
\def\uvou{{\underline{\mathbf{u}}^{\rm out}}}

\def\uv{\underline{\mathbf{u}}}

\def\ucv{\underline{\mathbf{u}}_K}

\def\grad{{\boldsymbol{\nabla}}}

\def\massmatrix{{\mathbf{M}}}

\def\stifmatrix{{\mathbf{A}}}
\def\condnum#1{\kappa(#1)}

\def\eigmin{{\lambda^{-}}}
\def\eigmax{{\lambda^{+}}}

\def\cmat{{\mathbf{C}}}

\def\cmatcref{{\mathbf{C}_{\aggrref}}}
\def\ext#1{\mathcal{E}(#1)}

\def\A{{\mathcal{A}}}

\def\Ac{{{\mathcal{A}}}}

\def\x{\boldsymbol{x}}

\def\next#1{\mathbf{E}#1}

\def\n{{\boldsymbol{n}}}

\def\Vc{{\FV\fespst}}

\newcommand{\gradient}{\boldsymbol{\nabla}}
\newcommand{\normal}{\boldsymbol{n}}

\newcommand{\FV}{}

\graphicspath{{Figures/}}

\begin{document}

\title[Aggregated unfitted finite element method]{The aggregated unfitted finite element method for elliptic problems}

\author[S. Badia]{Santiago Badia}

\author[F. Verdugo]{Francesc Verdugo}

\author[A. F. Mart\'in]{Alberto F. Mart\'in}

\address{Department of Civil and Environmental Engineering. Universitat Polit\`ecnica de Catalunya, Jordi Girona 1-3, Edifici C1, 08034 Barcelona, Spain.}
\address{CIMNE – Centre Internacional de M\`etodes Num\`erics en 
Enginyeria, Parc Mediterrani de la Tecnologia, UPC, Esteve Terradas 5, 08860 
Castelldefels, Spain.}

\thanks{SB gratefully acknowledges the support received from the Catalan Government through the ICREA Acad\`emia Research Program.  E-mails: {\tt sbadia@cimne.upc.edu} (SB), {\tt fverdugo@cimne.upc.edu} (FV), {\tt amartin@cimne.upc.edu} (AM)}

\date{\today}

\begin{abstract}
Unfitted finite element techniques are valuable tools in different applications where the generation of body-fitted meshes
is difficult. However, these techniques are prone to severe ill conditioning problems that obstruct the efficient use of iterative Krylov methods and, in consequence, hinders the practical usage of unfitted methods for realistic large scale applications. In this work, we present a technique that addresses such conditioning problems by constructing enhanced finite element spaces based on a cell aggregation technique. The presented method, called \emph{aggregated unfitted finite element method},  is easy to implement, and can be used, in contrast to previous works, in Galerkin approximations of coercive problems with conforming Lagrangian finite element spaces. The mathematical analysis of the new method states that the condition number of the resulting linear system matrix scales as in standard finite elements for body-fitted meshes, without being affected by small cut cells, and that the method leads to the optimal finite element convergence order. These theoretical results are confirmed with 2D and 3D numerical experiments.
\end{abstract}

\maketitle

%\noindent{\bf 2010 Mathematics Subject Classification:} 35Q30; 65N30; 76N10.


\tableofcontents

\section{Introduction} \label{sec:int}

Unfitted \ac{fe} techniques are specially appealing when the generation of \emph{body-fitted} meshes is difficult.  They are helpful in a number of contexts including  multi-phase and multi-physics applications with moving interfaces (e.g., fracture mechanics, fluid-structure interaction \cite{badia_fluidstructure_2008}, or free surface flows), or in situations in which one wants to avoid the generation of body-fitted meshes to simplify as far as possible the pre-processing steps
(e.g., shape or topology optimization frameworks, medical simulations based on CT-scan data, or
parallel large-scale simulations).  In addition, the huge success of isogeometrical analysis
(spline-based discretization) and the severe limitations of this approach in complex 3D geometries
will probably increase the interest of unfitted methods in the near future
\cite{kamensky_immersogeometric_2015}.  Unfitted \ac{fe} methods have been named in different ways. When
designed for capturing interfaces, they are usually denoted as eXtended \ac{fe} methods (XFEM)
\cite{belytschko_arbitrary_2001}, whereas they are usually denoted as embedded (or immersed)
boundary methods, when the motivation is to simulate a problem using a (usually simple Cartesian)
background mesh (see, e.g.,
\cite{burman_cutfem:_2015,mittal_immersed_2005,schillinger_finite_2014}).  

Yet useful, unfitted \ac{fe} methods have known drawbacks. They pose problems to numerical integration,
imposition of Dirichlet boundary conditions, and lead to ill conditioning problems. Whereas
different techniques have been proposed in the literature to address the issues related with
numerical integration (see, e.g., \cite{sudhakar_accurate_2014}) and the imposition of Dirichlet
boundary conditions (see, e.g., \cite{parvizian_finite_2007}), the conditioning problems are one of
the main showstoppers still today for the successful use of this type of methods in realistic large
scale applications. For most of the unfitted \ac{fe} techniques, the condition number of the discrete
linear system does not only depend on the characteristic element size of the background mesh, but
also on the characteristic size of the cut cells, which can be arbitrary small and have
arbitrarily high aspect ratios. This is an important problem. At large scales, linear systems are
solved with iterative Krylov sub-space methods \cite{saad_iterative_2003} in combination with
scalable preconditioners. Unfortunately, the well known scalable preconditioners based on
(algebraic) multigrid \cite{briggs_multigrid_2000} or multi-level domain decomposition
\cite{badia_multilevel_2016} are mainly designed for body-fitted meshes and cannot readily deal with
cut cells. Different preconditioners for unfitted \ac{fe} methods have been recently proposed, but they are
mainly serial non-scalable algorithms  (see, e.g.,
\cite{menk_robust_2011,berger-vergiat_inexact_2012,hiriyur_quasi-algebraic_2012,de_prenter_condition_2017}).
Recently, a robust domain decomposition preconditioner able to deal with cut cells has been proposed in
\cite{badia_robust_2017}. Even though this method has proven to be scalable in some complex 3D
examples, it is based on heuristic considerations without a complete mathematical analysis and its application to second (and higher) order \acp{fe} is involved. This lack of preconditioners for unfitted \acp{fe} can be
addressed with enhanced formulations that provide well-posed discrete systems independently of the
size of the cut cells. Once the conditioning problems related to cut cells are addressed, the application of standard
preconditioners for body-fitted meshes to the unfitted case is strongly simplified, opening the door to large-scale computations.

The main goal of this work is to develop such an enhanced unfitted \ac{fe} formulation that fixes the
problems associated with cut cells. The goal is to achieve condition numbers that scale only with
the element size of the background mesh in the same way as in standard \ac{fe} methods for
body-fitted meshes. Our purpose is to implement it in FEMPAR, our in-house large scale \ac{fe} code
\cite{badia_fempar:_2017}. Since FEMPAR is a parallel multi-physics multi-scale code that includes
different continuous and discontinuous \ac{fe} formulations and several element types, it is crucial for
us that the novel formulation fulfills the following additional properties: 1) It should be general
enough to be applied to several problem types, 2) it should deal with both continuous and
discontinuous \ac{fe} formulations, 3)  it should deal with high order interpolations, and 4) it should
be easily implemented in an existing parallel \ac{fe} package.

To our best knowledge, none of the existing unfitted \ac{fe}
formulations fulfill these requirements simultaneously. For instance,
one can consider the \emph{ghost penalty} formulation used in the
CutFEM method \cite{burman_cutfem:_2015,burman_fictitious_2012}
However, it leads to a weakly non-consistent algorithm, and it
requires to compute high order derivatives on faces for high order
\acp{fe}, which are not at our disposal in general \ac{fe} codes and
are expensive to compute, certainly complicating the implementation of
the methods and harming code performance. Alternatively, for
finite volume and \ac{dg} formulations, one can consider the so-called
\emph{cell aggregation} (or agglomeration) techniques \cite{helzel_high-resolution_2005,kummer_extended_2013}. E.g., for \ac{dg} formulations, the
idea is simple: cells with the \emph{small cut cell problem}, i.e., the
ratio between the volume of the cell inside the physical domain and
the total cell volume is close to zero, are merged with neighbor full
cells forming aggregates. A new polynomial space is defined in each
aggregate that replaces the local \ac{fe} spaces of all cells merged
in it. This process fixes the conditioning problems, since the support
of the newly defined shape functions is at least the volume of a full
cell. Even though this idea is simple and general enough to deal with
different problem types and high order interpolations, the resulting
discrete spaces are such that the enforcement of continuity through
appropriate local-to-global \acp{dof} numbering, as in standard \ac{fe}
codes (see, e.g., \cite{badia_fempar:_2017}), is not possible, limiting their usage to discontinuous
Galerkin or finite volume formulations. Up to our best knowledge, there
is no variant of cell agglomeration currently proposed in the
literature producing conforming \ac{fe} spaces, which could be used
for classical continuous Galerkin formulations. It is the purpose of this work.

In this article, we present an alternative cell aggregation technique that can be used for both
continuous and discontinuous formulations, the \emph{aggregated unfitted \ac{fe} method}. We start with the usual (conforming) Lagrangian \ac{fe} space that includes cut cells, which
is known to lead to conditioning problems. The main idea is to  eliminate
from this space all the potentially problematic \acp{dof} by introducing a set of judiciously
defined constraints. These constraints are introduced using information provided by the cell
aggregates, without altering the conformity of the original \ac{fe} space. Alternatively, the method can be understood as an extension operator from the interior (\emph{well-posed}) \ac{fe} space that only involves interior cells to a larger \ac{fe} space that includes cut cells and covers the whole physical domain. Discontinuous spaces can
also be generated as a particular case of this procedure, which makes the method compatible also
with \ac{dg} formulations. In contrast to previous works, we also include a
detailed mathematical analysis of the method, in terms of well-posedness, condition number estimates, and \emph{a priori} error estimates.  For elliptic problems, we mathematically prove that 1) the
method leads to condition numbers that are independent from small cut cells, 2) the condition
numbers scale with the size of the background mesh as in the standard \ac{fe} method, 3) the penalty parameter of
Nitsche's method required for stability purposes is bounded above, and 4) the optimal \ac{fe} convergence order is recovered. These
theoretical results are confirmed with 2D and 3D numerical experiments using the Poisson equation as
a model problem. 

The outline of the article is as follows. In Section \ref{sec:bac_mesh}, we
introduce our embedded boundary setup and the strategy to build the cell aggregates. In Section
\ref{sec:agg_unf}, we describe the construction of the novel \ac{fe} spaces based on the cell
aggregates. In Section \ref{sec_app_ell}, we introduce our elliptic model problem. The numerical
analysis of the method is carried out in Section \ref{sec:sta_ana}. Finally, we present a complete set of numerical experiments in Section \ref{sec:num_exp} and
draw some conclusions in Section~\ref{sec:concl}.

\section{\FV Embedded boundary setup and cell aggregation} \label{sec:bac_mesh}

% \task{Define the original partition, the interior partition}

Let $\Dom \subset \mathbb{R}^d$ be an open bounded polygonal domain, with $d\in\{2,3\}$ the number of spatial dimensions. For the sake of simplicity and without loss of generality, we consider in the numerical experiments below that the domain boundary is defined as the zero level-set of a given scalar function $\lset^\mathrm{ls}$, namely $\partial\Omega\doteq\{x\in\mathbb{R}^d:\lset^\mathrm{ls}(x)=0\}$.\footnote{Analogous assumption have to be made for body-fitted methods.} We note that the problem geometry could be described using 3D CAD data instead of level-set functions, by providing techniques to compute the intersection between cell edges and surfaces (see, e.g., \cite{marco_exact_2015}). In any case, the way the geometry is handled does not affect the following exposition. Like in any other embedded boundary method, we build the computational mesh by introducing an \emph{artificial} domain $\Domart$ such that it has a simple geometry that is easy to mesh using Cartesian grids and it includes the \emph{physical} domain $\Omega \subset\Domart$  (see Fig. ~\ref{fig:immersed-setup-a}).

\begin{figure}[ht!]
  \centering
  \begin{subfigure}{0.24\textwidth}
    \centering
    \includegraphics[width=0.9\textwidth]{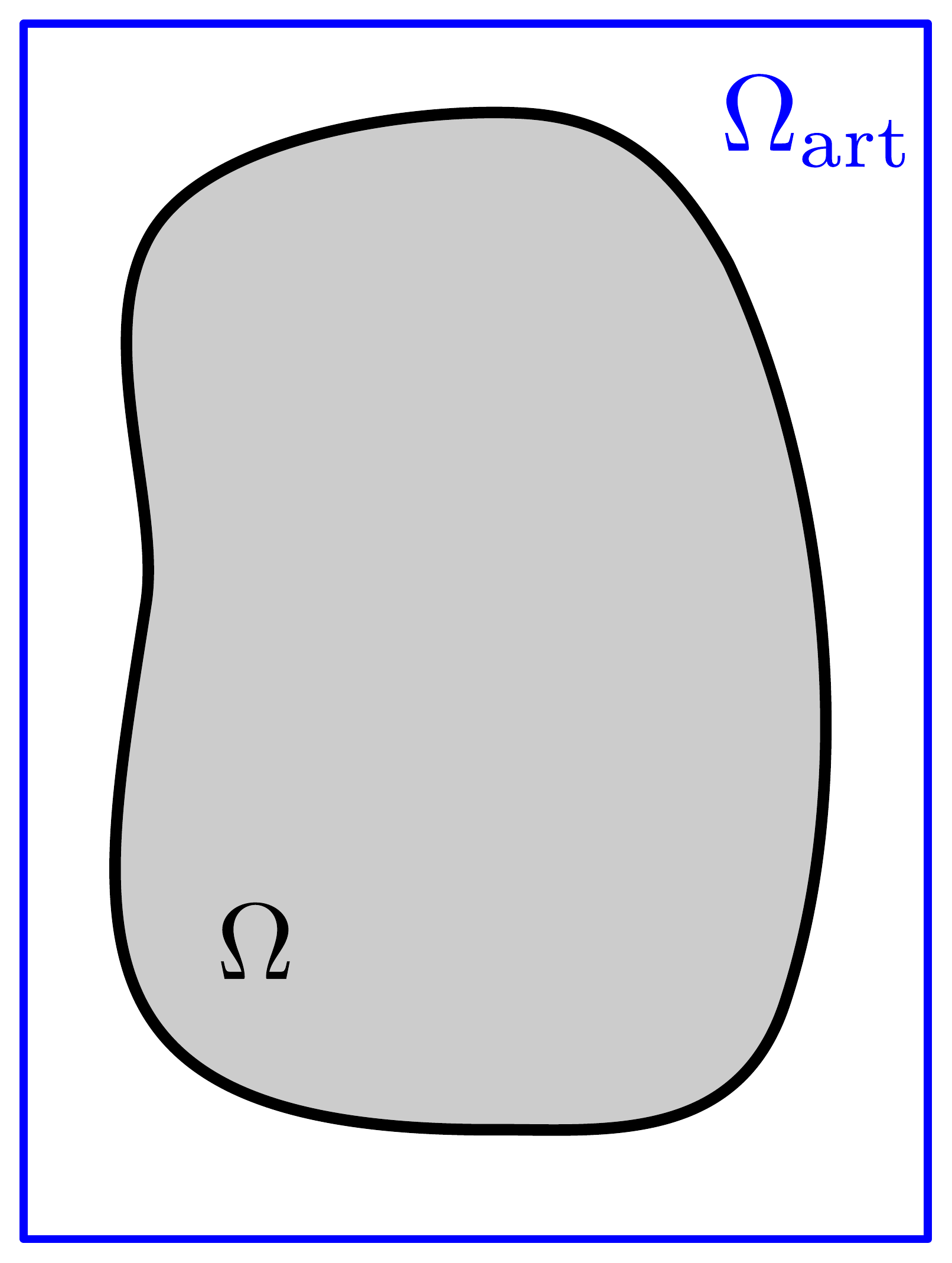}
    \caption{}
    \label{fig:immersed-setup-a}
  \end{subfigure}
  \begin{subfigure}{0.24\textwidth}
    \centering
    \includegraphics[width=0.9\textwidth]{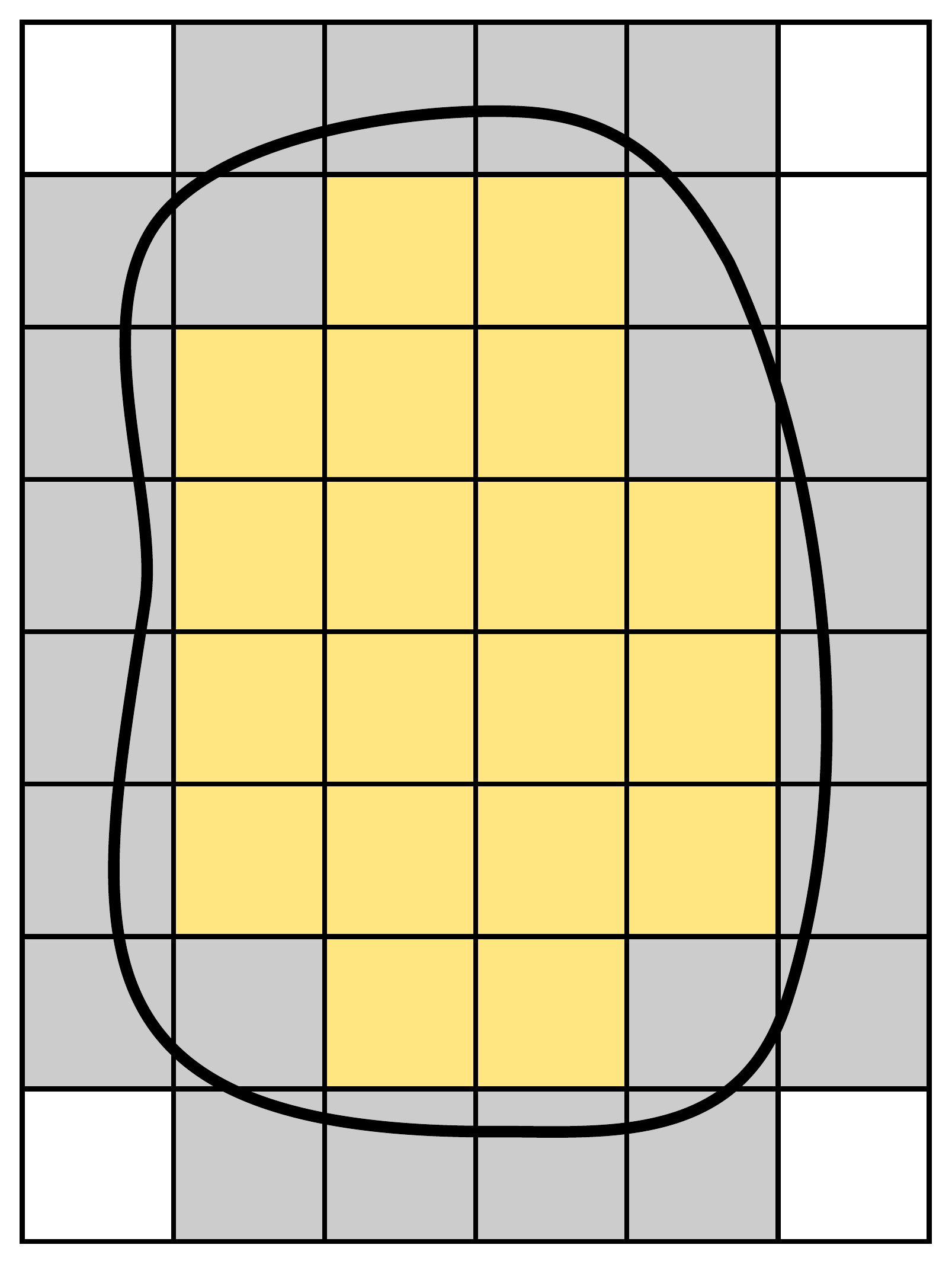}
    \caption{}
    \label{fig:immersed-setup-b}
  \end{subfigure}
    \begin{subfigure}{0.1\textwidth}
    %\centering
    \begin{tabular}{l}
    \tikz{\draw[fill=myellow]  (0,0) rectangle (1.4em,1.4em);} \emph{internal} cells
    \\
    \tikz{\draw[fill=gray20]  (0,0) rectangle (1.4em,1.4em);} \emph{cut} cells
    \\
    \tikz{\draw  (0,0) rectangle (1.4em,1.4em);} \emph{external} cells
    \end{tabular}
  \end{subfigure}
  \caption{Embedded boundary setup.}
  \label{fig:immersed-setup}
\end{figure}

Let us construct a partition of $\Domart$ into \emph{cells}, represented by $\meshart$, with characteristic cell size $h$. We are interested in $\meshart$ being a Cartesian mesh into hexahedra for $d=3$ or quadrilaterals for $d = 2$, even though unstructured n-simplex background meshes can also be considered. Cells in $\meshart$ can be classified as follows: a cell $\cell \in \meshart$ such that $\cell \subset \Omega$ is an \emph{internal cell}; if $\cell \cap \Omega = \emptyset$, $\cell$ is an \emph{external cell}; otherwise, $\cell$ is a \emph{cut cell} (see Fig.~\ref{fig:immersed-setup-b}).
The set of interior (resp., external and cut) cells is represented with $\meshin$ and its union $\Domin \subset \Omega$ (resp., $(\meshext,\Domext)$ and $(\meshcut, \Domcut))$. Furthermore, we define the set of \emph{active cells} as $\meshact \doteq \meshin \cup \meshcut$ and its union $\Domex$. In the numerical analysis, we assume that the background mesh is \emph{quasi-uniform} (see, e.g., \cite[p.107]{brenner_mathematical_2010}) to reduce technicalities, and define a characteristic mesh size $\h$. The maximum element size is denoted with $\h_{\rm max}$.

%\task{Define the aggregation partition and the root cell, bounding box, aggregation ratio}

%\task{Define a one-to-one mapping between vefs and cells}

We can also consider non-overlapping cell aggregates $\aggr$ composed of cut cells and one interior cell $\cell$ such that the aggregate is connected, using, e.g., the strategy described in Algorithm \ref{algo:aggr_scheme}. It leads to another partition $\meshag$ defined by the aggregations of cells in $\meshex$; interior cells that do not belong to any aggregate remain the same. By construction of Algorithm \ref{algo:aggr_scheme}, there is only one interior cell per aggregate, denoted as the \emph{root cell} of the aggregate, and every cut cell belongs to one and only one aggregate. For a cut cell, we define its root cell as the root of the only aggregate that contains
the cut cell. The root of an interior cell is the cell itself. Thus, there is a one-to-one mapping between aggregates (including interior cells) $\aggr \in\meshag$ and the root cut cell $K \in \meshin$. As a result, we can use the same index for the aggregate and the root cell. We build the aggregates in $\meshag$ with Algorithm \ref{algo:aggr_scheme}. In any case, other aggregation algorithms could be considered, e.g., touching in the first step of the algorithm not only the interior cells, but also cut cells \emph{without} the small cut cell problem. It can be implemented by defining the quantity
$\eta_\cell \doteq  \frac{| \cell \cap \Omega|}{ |\cell|}$ and touch in the first step not only the interior cells but also any cut cell with $\eta_\cell > \eta_0 > 0$ for a fixed value $\eta_0$.
 
%\TBD{Mirar que el paragraf que segueix és realment necessari per alguna secció de les que segueix.}

%For a given cell $\cell \in \meshin$, we can define the \emph{bounding n-cube} of $\cell$, represented by $\boxc$, as the smallest n-cube $\boxc$ such that $\aggr \subseteq \boxc$. Given two cells in the mesh, the corresponding bounding  n-cubes can overlap. We also define the \emph{aggregation ratio} $\gamma_{\cell} = (|\boxc|/|\cell|)^{\frac{1}{d}}$. We have $\boxc = \gamma_{\cell} \cell + \hc \i_\cell$, i.e., after considering a scaling and translation of the root cell $\cell$. %\Quest{ Is $\i_\cell$ an arbitrary vector? or is it a generic vector of the canonical basis?}

%Finally, let us also consider a one-to-one map $\own{\cdot}$ between every \ac{vef} of the mesh and one of the cells that contain it, e.g., using the strategy in \TBD{}. 

%\task{Minimization process in defining the aggregation, also for cell map, by @fverdugo}

%\section{Aggregation scheme}
%\label{sec:aggr_scheme}

%We can consider different aggregation schemes. The one proposed here
%is just an example. We mark all the external cells as untouched. Next,
%we traverse all the external cells. For every untouched external cell,
%if there is at least one full cell or touched external cell connected
%to it through a facet, we aggregate the external cell to the full cell
%with minimum global index. If not all the external
%cells have been touched yet, we proceed again.

\begin{method}[Cell aggregation scheme]\

\begin{enumerate}
\item Mark all interior cells as touched and all cut cells as untouched.
\item For each untouched cell,  if there is at least one touched cell connected
to it through a facet  $F$ such that $F \cap \Omega \neq \emptyset$, we aggregate the cell to the touched cell belonging to the aggregate containing the closest interior cell. If more than one touched cell fulfills this requirement, we choose one arbitrarily, e.g., the one with smaller global id.  
\item Mark as touched all the cells aggregated in 2. 
\item Repeat 2. and 3. until all cells are aggregated.
\end{enumerate}
\label{algo:aggr_scheme}
\end{method}

\begin{figure}[ht!]
  \centering
  \begin{subfigure}{0.99\textwidth}
    \centering
    \begin{small}
      \begin{tabular}{llll}
         \tikz{\fill[fill=myellow]  (0,0) rectangle (1.4em,1.4em);} touched 
         &
         \tikz{\fill[fill=gray20]    (0,0) rectangle (1.4em,1.4em);} untouched
         &
         \tikz{ \draw[line width=0.5pt] (0,0) -- (2em,0);} Aggregates' boundary
         &
         \tikz{ \draw[line width=2pt] (0,0) -- (2em,0);} $\partial \Omega$
      \end{tabular}
    \end{small}
  \end{subfigure}
  \par
  \begin{subfigure}{0.24\textwidth}
    \centering
    \includegraphics[width=0.9\textwidth]{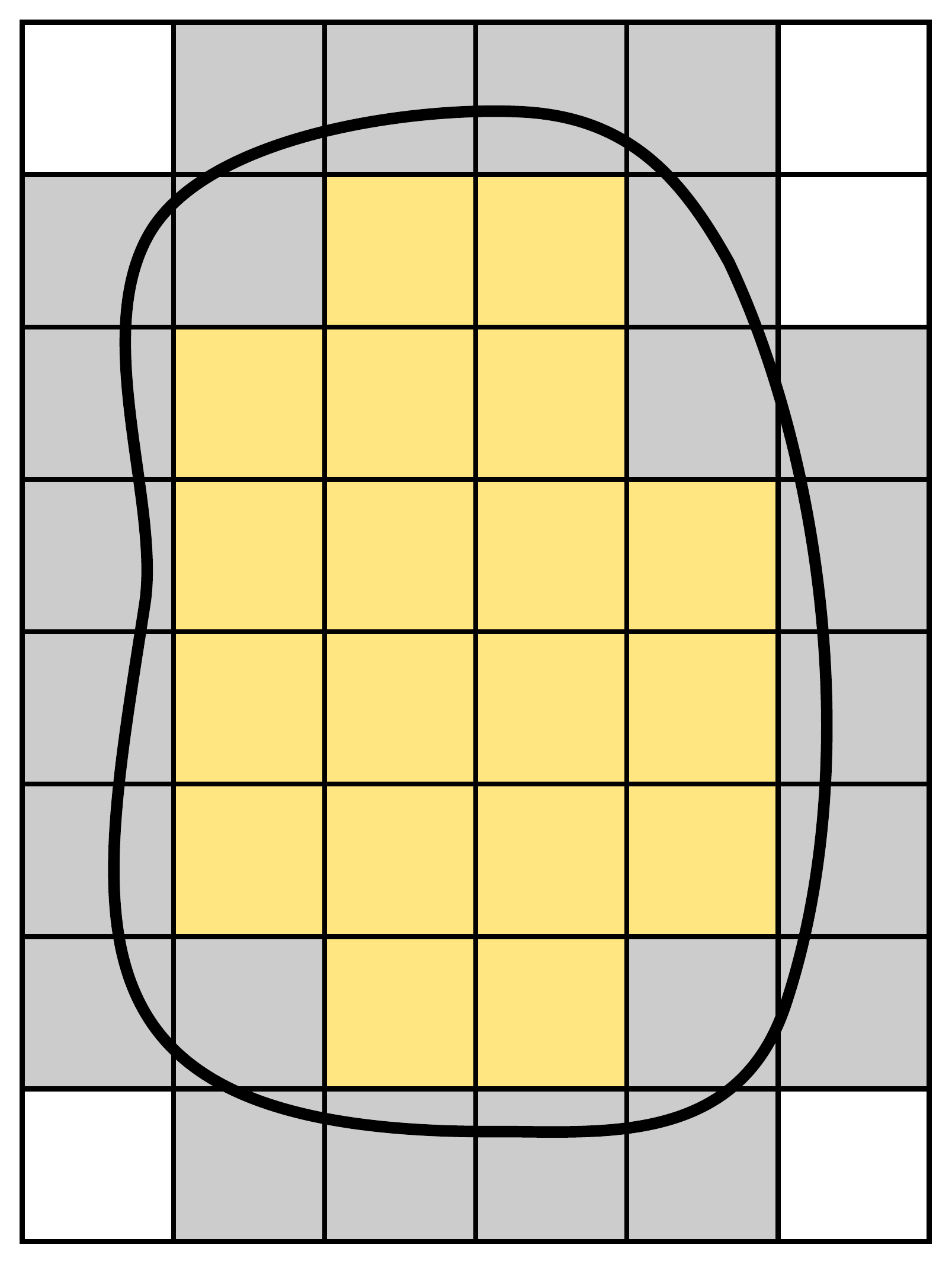}
    \caption{Step 1}
    \label{fig:aggr-steps-a}
  \end{subfigure}
  \begin{subfigure}{0.24\textwidth}
    \centering
    \includegraphics[width=0.9\textwidth]{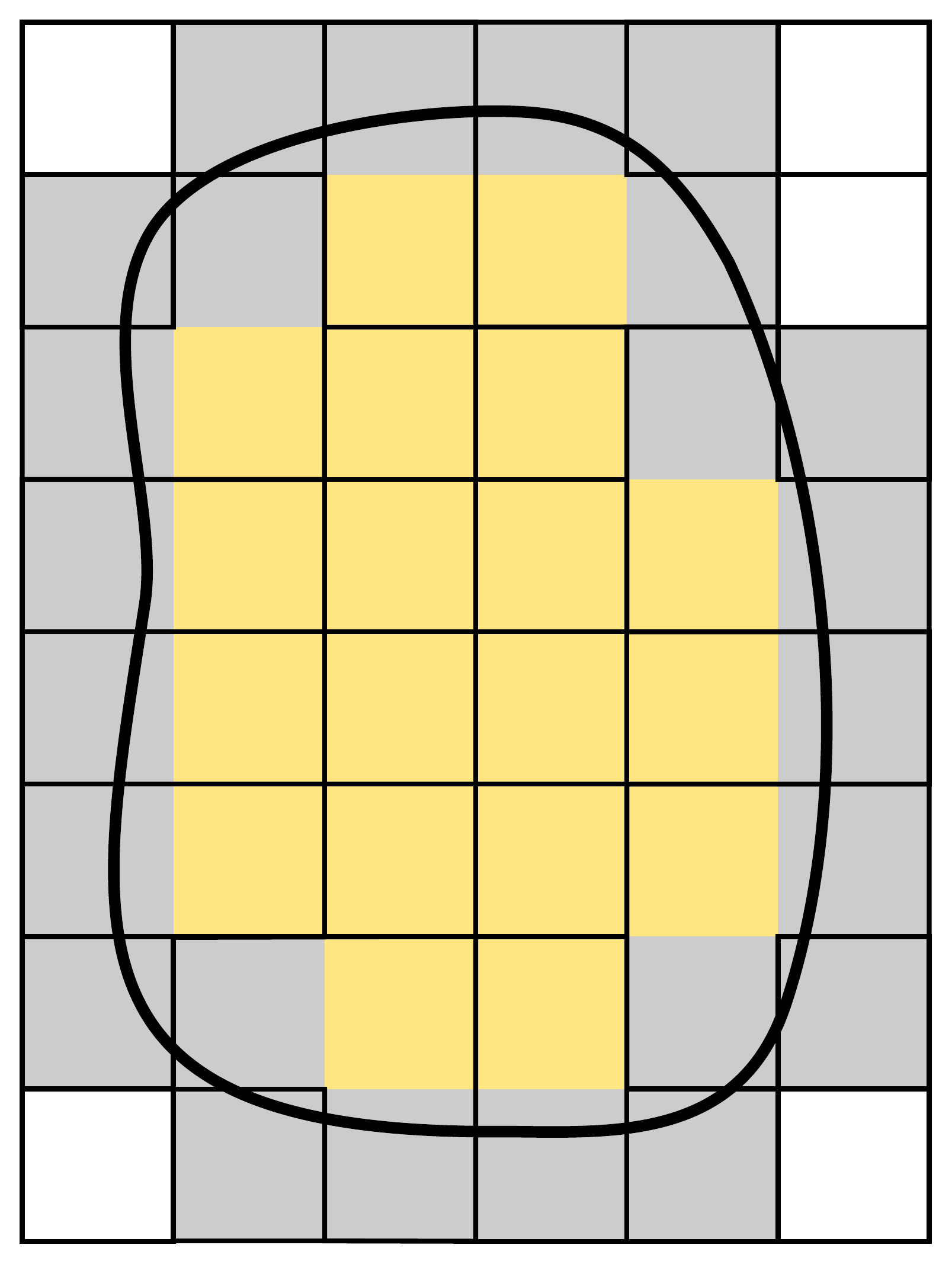}
    \caption{Step 2}
    \label{fig:aggr-steps-b}
  \end{subfigure}
  \begin{subfigure}{0.24\textwidth}
    \centering
    \includegraphics[width=0.9\textwidth]{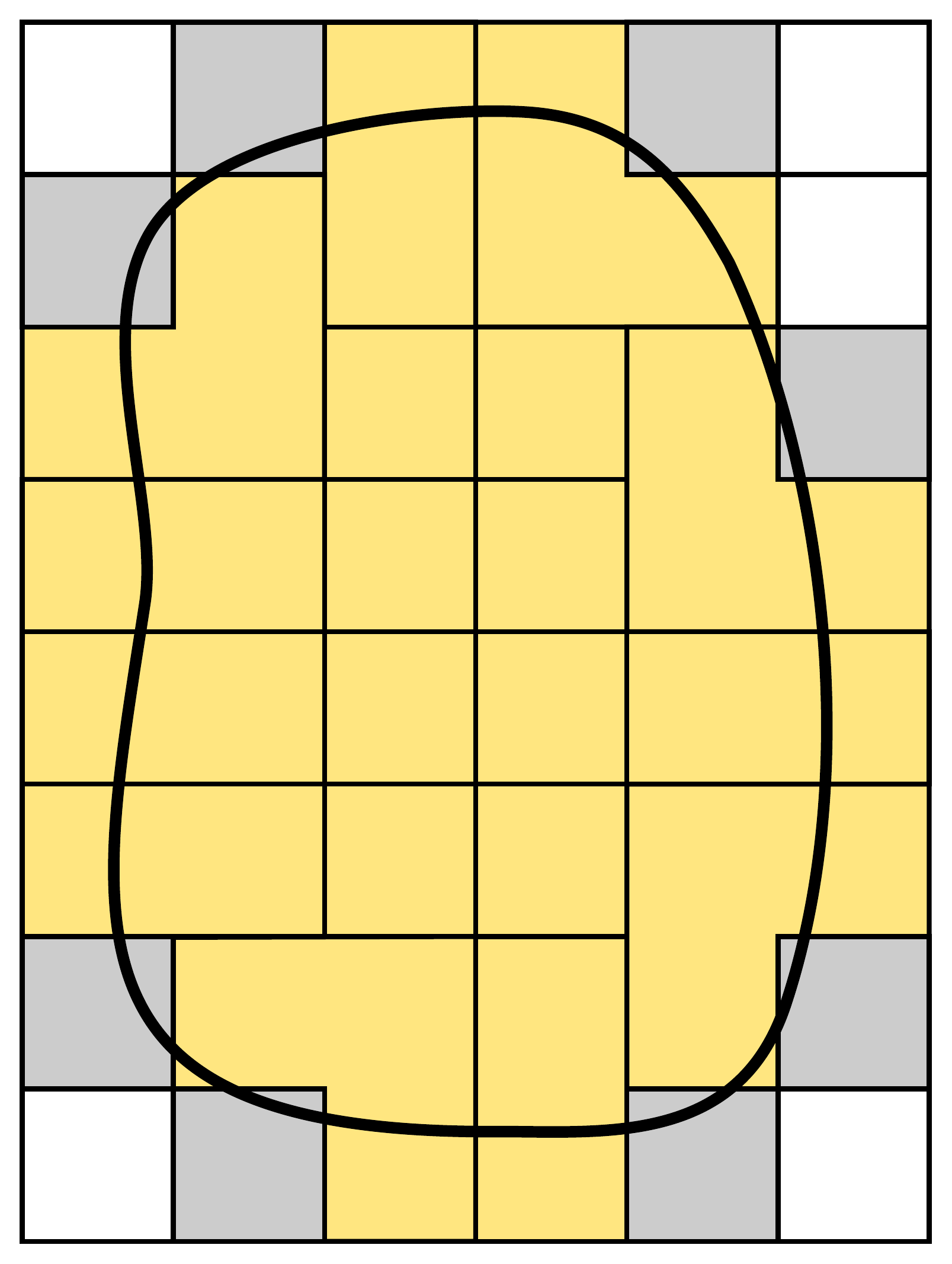}
    \caption{Step 3}
    \label{fig:aggr-steps-c}
  \end{subfigure}
  \begin{subfigure}{0.24\textwidth}
    \centering
    \includegraphics[width=0.9\textwidth]{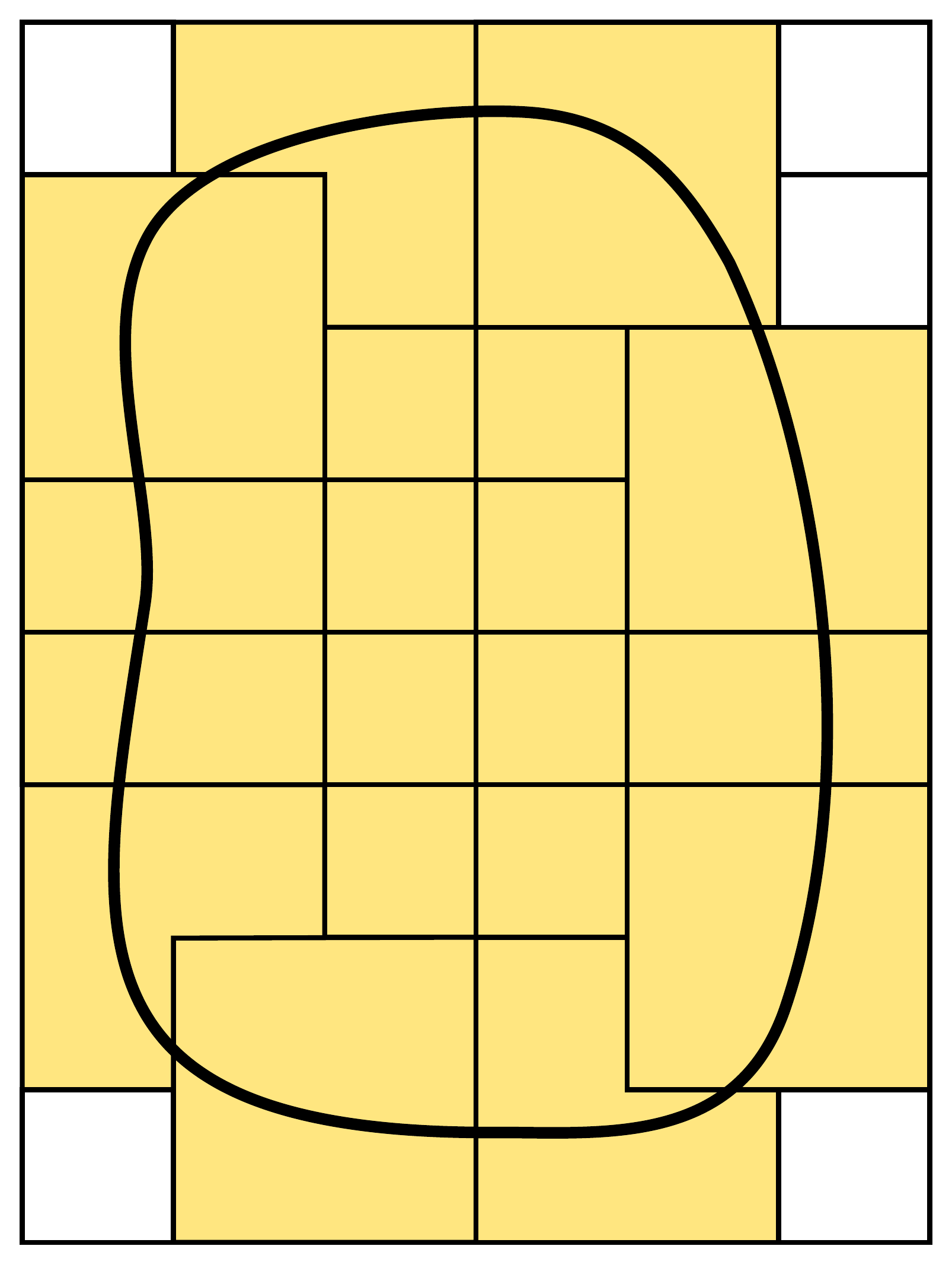}
    \caption{Step 4}
    \label{fig:aggr-steps-d}
  \end{subfigure}
  \caption{Illustration of the cell aggregation scheme defined in Algorithm \ref{algo:aggr_scheme}.}
  \label{fig:aggr-steps}
\end{figure}

Fig. \ref{fig:aggr-steps} shows an illustration of each step in Algorithm \ref{fig:aggr-steps}. The black thin lines represent the boundaries of the aggregates. Note that from step 1 to step 2, some of the lines between adjacent cells are removed, meaning that the two adjacent cells have been merged in the same aggregate. The aggregation schemes can be easily applied to arbitrary spatial dimensions. As an illustrative example,  Fig. \ref{fig:aggr-3d} shows some of the aggregates obtained for a complex 3D domain.

\begin{figure}[ht!]
 \centering
 \includegraphics[width=0.7\textwidth]{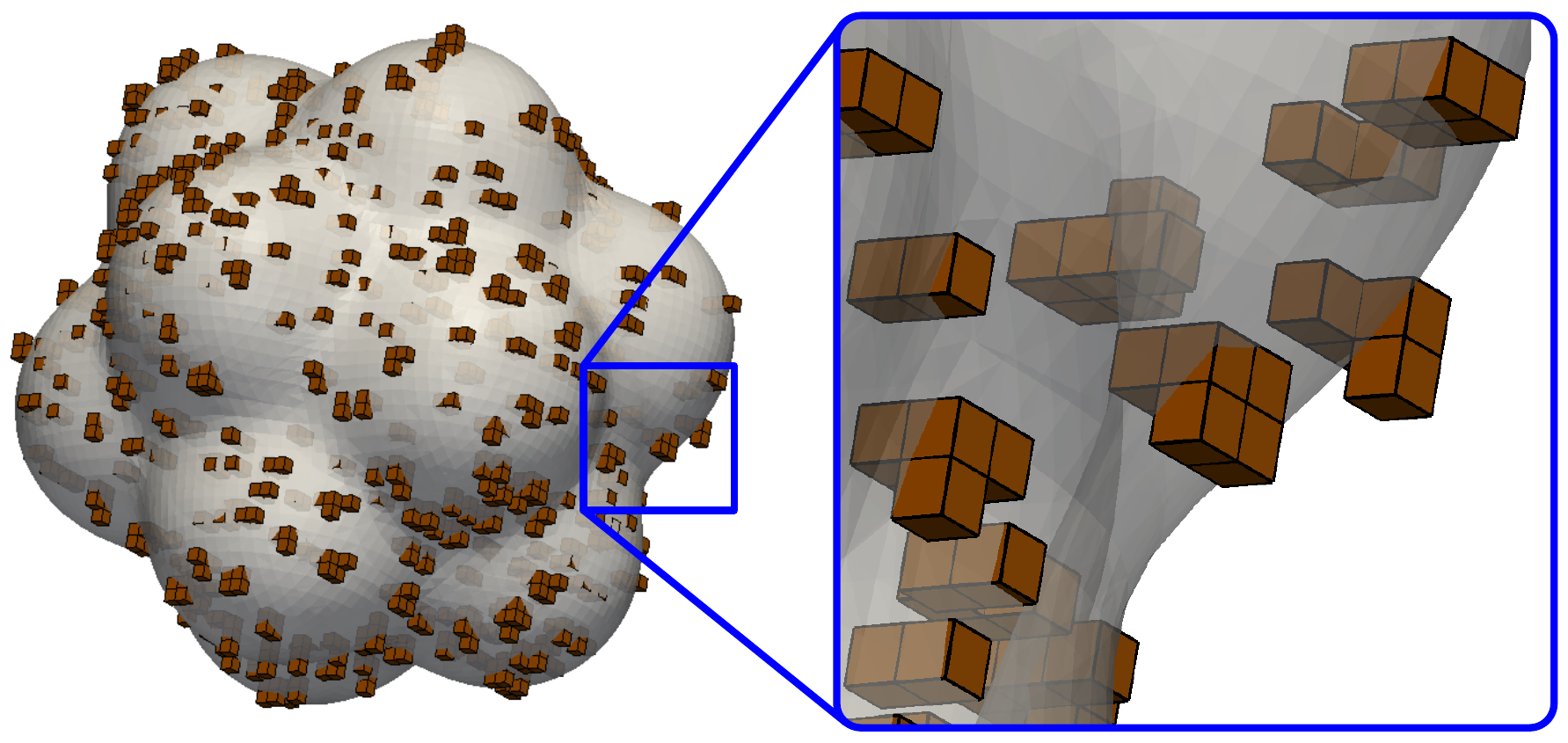}
 \caption{3D aggregates.} %For the sake of clarity, only some aggregates are visualized.}
 \label{fig:aggr-3d}
\end{figure}

In the forthcoming sections, we need an upper bound of the size of the aggregates generated with Algorithm \ref{algo:aggr_scheme}. To this end, let us consider the next lemma.

\begin{lemma}\label{lemma:aggrsize}
  Assume that from any cut cell $\cell_0 \in \meshex$ there is a cell path $\{ \cell_0, \cell_1, \ldots, \cell_n \}$ that satisfies: 1) two consecutive cells share a facet $F$ such that $F \cap \Omega \neq \emptyset$; 2) $\cell_n$ is an interior cell; 3) $n \leq \gamma_{\rm max}$, where $\gamma_{\rm max}$ is a fixed integer. Then, the maximum aggregate size is at most $(2\gamma_{\rm max}+1)\h_{\rm max}$.
  \end{lemma}
\begin{proof}
By construction, an aggregate can grow at most at a rate of one layer of elements per each   iteration.  Thus, after $n$ iterations the aggregate size will be at most $(2n+1)h_{\rm max}$  considering that the aggregate can potentially grow in all spatial directions. 
%(In order this to be true, it is important to mark the cells as touched AFTER ALL cells have been aggregated in step 2.) 
It is obvious to see that the aggregation scheme finishes at most after $\gamma_{\rm  max}$ iterations. Thus, the aggregate size will be less or equal than $(2\gamma_{\rm max}+1) \h_{\rm max}$.
\end{proof}

From Lemma \ref{lemma:aggrsize}, it follows that the aggregate size will be bounded if so is the value of $\gamma_{\rm max}$. In what follows, we assume that $\gamma_{\rm max}$ is fixed, e.g., eliminating any cut cell that would violate property 3) in Lemma \ref{lemma:aggrsize}. %\footnote{In practice, large values of $\gamma_{\rm max}$ can only appear when dealing with large geometrical details of thickness below the mesh resolution. In these situations, more refined background meshes would be needed to properly capture these geometrical details.}
 One shall assume that each cut cell shares at least one corner with an interior cell (this is usually true if the grid is fine enough to capture the geometry). In this situation, we can easily see that $\gamma_{\rm max}=2$ for 2D and $\gamma_{\rm max}=3$ for 3D. Then, by Lemma \ref{lemma:aggrsize}, the aggregate size is at most $5h_{\rm max}$ in 2D and $7h_{\rm max}$ in 3D.  Even though it is not used in the proof of Lemma \ref{lemma:aggrsize}, the fact that we aggregate cut cells to the touched cells belonging to the aggregate containing the closest interior cell (see step 2 in Algorithm \ref{algo:aggr_scheme}) contributes to further reduce the aggregate size. Indeed, the actual size of the aggregates generated in the numerical examples (cf. Section \ref{sec:num_exp}) is much lower than the predicted by these theoretical bounds. In 2D, the aggregate size tends to $2h_{\rm max}$ as the mesh is refined, whereas it tends to $3h_{\rm max}$ in the 3D case. This shows that the aggregation scheme produces relative small aggregates in the numerical experiments. %{\color{red} Comprovar que sigui així pels nous càlculs. Crec que la cota en 3D tb es pot reduir a $5h$ pel fet que estem aggregant sempre cap a l'element plè més proper.  Fent alguna hipòtesi més, potser podríem arribar a $3h$.}

\section{Aggregated unfitted Lagrangian finite element spaces} \label{sec:agg_unf}

%\task{Define the interior fe space}

{\FV 

Our goal is to define a \ac{fe} space using the cell aggregates introduced above. To this end, we need to introduce some notation. In the case of n-simplex meshes, we define the local \ac{fe} space $V(\cell) \doteq \mathcal{P}_q(\cell)$, i.e., the space of polynomials of order less or equal to $q$ in the variables $x_1,\ldots,x_d$. For n-cube meshes, $V(\cell) \doteq \mathcal{Q}_q(\cell)$, i.e.,  the space of polynomials that are of degree less or equal to $k$ with respect to each variable $x_1, \ldots, x_d$.  In this work, we consider that the polynomial order $q$ is the same for all the cells in the mesh. We restrict ourselves to Lagrangian \ac{fe} methods. Thus, the basis for $V(\cell)$ is the Lagrangian basis (of order $q$) on $\cell$. 
We denote by $\nodes{\cell}$ the set of Lagrangian nodes of order $q$ of cell $\cell$. 
There is a one-to-one mapping between nodes $a \in \nodes{\cell}$  and shape functions $\shpf{a}(\x)$; 
it holds $\shpf{a}(\x^b) = \delta_{ab}$, where $\x^b$ are the space coordinates of node $b$.
%We can also define the \ac{vef} $F$ of $K$ that owns a node $a \in \nodes{\cell}$ as the one with minimum dimension that contains it.
%
}
%
%For interior cells, we can simply define the local finite element space using Ciarlet's definition as $V(\cell) \doteq (\cell,\mathcal{P}_{q(\cell)}(\cell),\Sigma_{q(\cell)}(\cell))$, and consider a dual basis $\Sigma$ for $V(\cell)'$. We restrict ourselves to Lagrangian \ac{fe} methods. Thus, the basis for $\mathcal{P}_{q(K)}(K)$ is the Lagrangian basis and the functionals basis $\Sigma$ are the Lagrangian nodal values (for order $q$) on $\cell$. We denote by $\nodes{\cell}$ the set of Lagrangian nodes of order $q$ of cell $\cell$. There is a one-to-one mapping between nodes $a \in \nodes{\cell}$ and moments $\sigma^a \in \Sigma$, and between nodes and shape functions $\shpf{a}(\x)$; it holds $\shpf{a}(\x^b) = \delta_{ab}$, where $\x^b$ are the space coordinates of node $b$. We can also define the \ac{vef} $F$ of $K$ that owns a moment as the one with minimum dimension that contains it. %For a given \ac{vef}, only the shape functions of moments belonging to the \ac{vef} itself or lower-dimensional \ac{vef}S contained in the \ac{vef} are non-zero on the \ac{vef}. These shape functions (restricted to the \ac{vef}) are a basis for $\mathcal{P}_{q(K)}(F)$.
%
%
We assume that there is a local-to-global \ac{dof} map such that the resulting global system is $\mathcal{C}^0$ continuous. This process can be elaborated for $hp$-adaptivity as well, but it is not the purpose of this work. 

{\FV

With this notation, we can introduce the \emph{active} \ac{fe} space associated with the active portion of the background mesh
$$
\fespex \doteq \{ v \in {\mathcal{C}^0}(\Domex) \, : \, v|_K \in
V(K), \, \hbox{for any} \, K \in \meshex \}.
$$
We could analogously define the \emph{interior} \ac{fe} space 
$$
\fespin \doteq \{ v \in {\mathcal{C}^0}(\Domin) \, : \, v|_K \in
V(K), \, \hbox{for any} \, K \in \meshin \}.
$$
The active \ac{fe} space $\fespex$ (see Fig. \ref{fig:def-spaces-c}) is the functional space typically used in unfitted \ac{fe} methods (see, e.g., \cite{badia_robust_2017,de_prenter_condition_2017,schillinger_finite_2014}).
 It is well known that $\fespex$ leads to arbitrary ill conditioned systems when integrating the \ac{fe} weak form on the physical domain $\Dom$ only (if no extra technique is used to remedy it).
It is obvious that the interior \ac{fe} space $\fespin$ (see Fig. \ref{fig:def-spaces-a}) is not affected by this problem, but it is not usable since it is not defined on the complete physical domain $\Dom$. Instead, we propose an {\FV alternative} space $\fespag$ that is defined on $\Domap$ but does not present the problems related to $\fespex$. We can define the set of nodes of $\fespin$ and $\fespex$ as $\nodesgin$ and $\nodesgex$, respectively (see Fig. \ref{fig:def-spaces}). We define the set of \emph{outer} nodes as $\nodesgou \doteq \nodesgex \setminus \nodesgin$ (marked with red crosses in Fig. \ref{fig:def-spaces-b}). The outer nodes are the ones that can lead to conditioning problems due to the small cut cell problem (see \Eq{finite-cell-estimate}). {\FV The space  $\fespag$} is defined taking as starting point $\fespex$, and adding judiciously defined constraints for the nodes in $\nodesgou$.

\begin{figure}[ht!]
  \centering
  \begin{subfigure}{0.24\textwidth}
    \centering
    \includegraphics[width=0.9\textwidth]{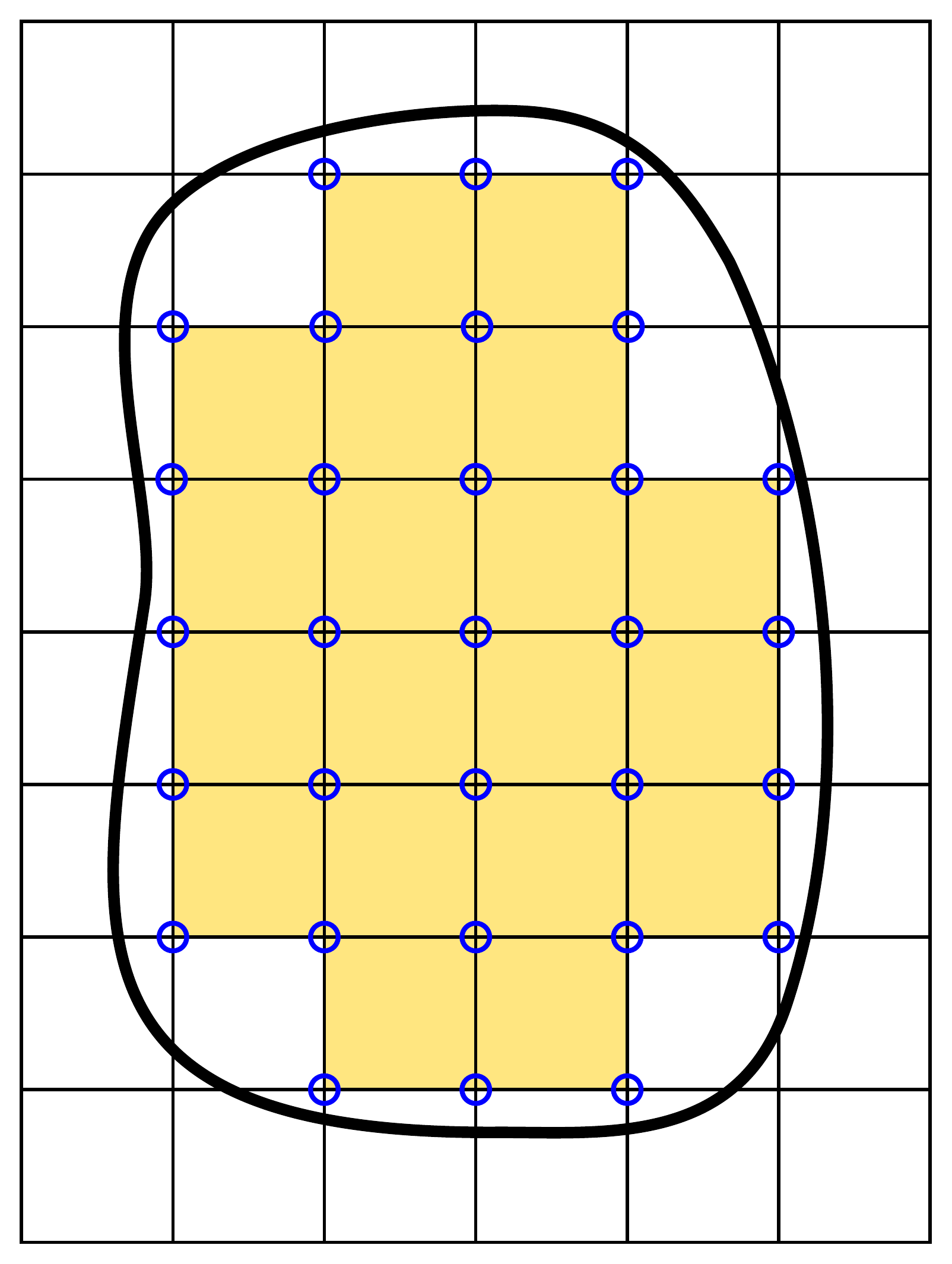}
    \caption{$\fespin$}
    \label{fig:def-spaces-a}
  \end{subfigure}
  \begin{subfigure}{0.24\textwidth}
    \centering
    \includegraphics[width=0.9\textwidth]{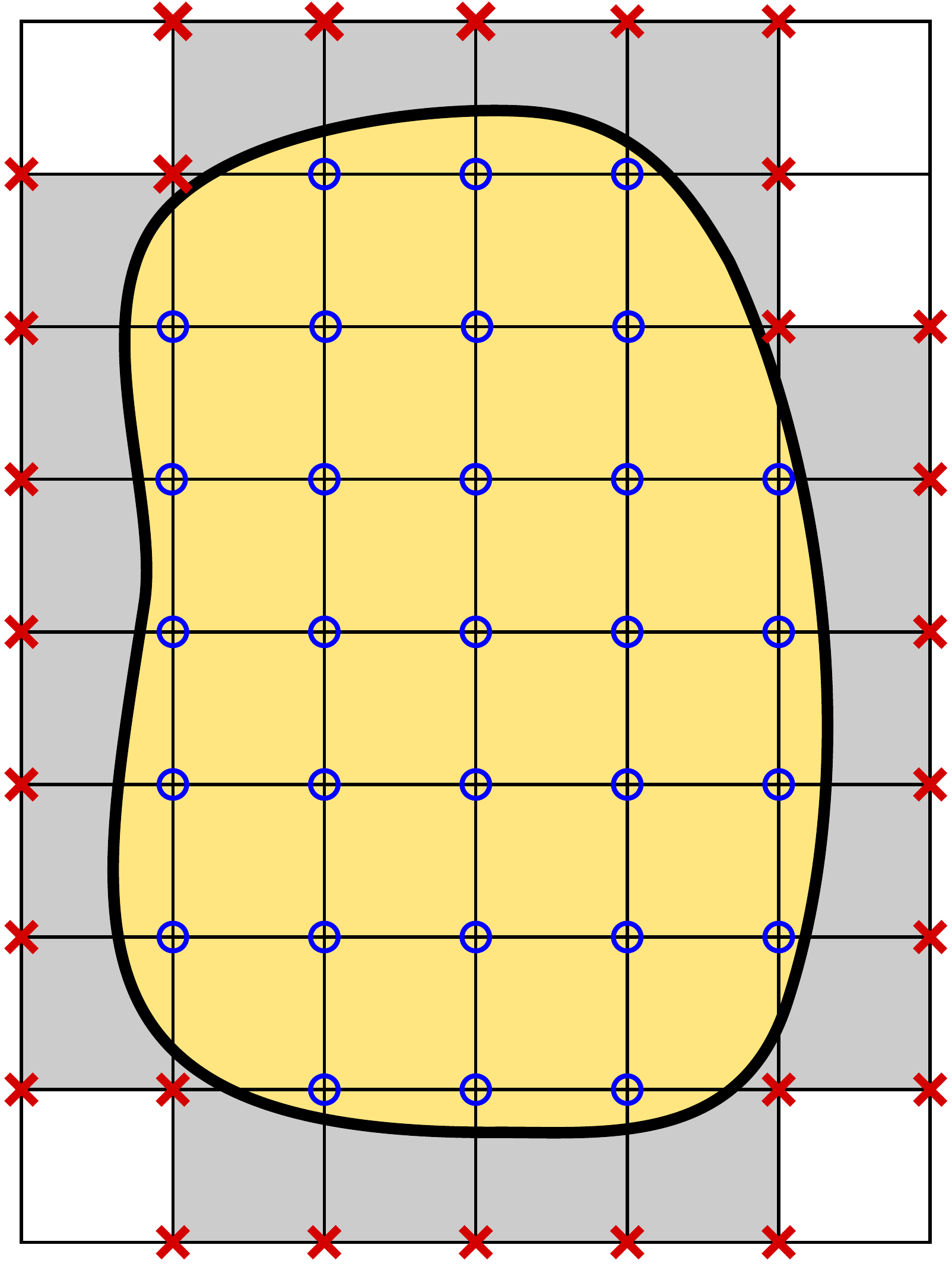}
    \caption{$\fespag$}
    \label{fig:def-spaces-b}
  \end{subfigure}
  \begin{subfigure}{0.24\textwidth}
    \centering
    \includegraphics[width=0.9\textwidth]{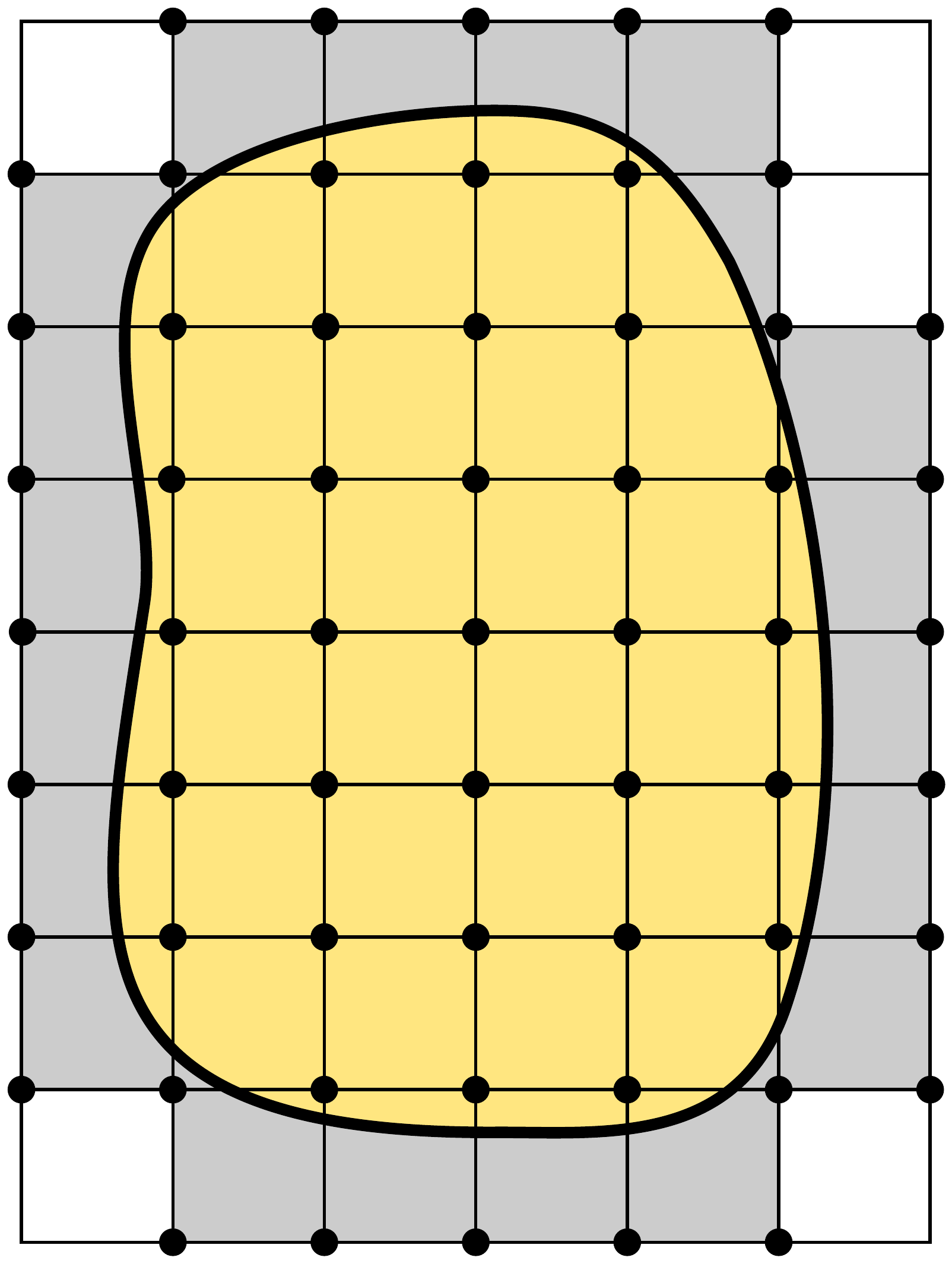}
    \caption{$\fespex$}
    \label{fig:def-spaces-c}
  \end{subfigure}
  \begin{subfigure}{0.2\textwidth}
    \begin{tabular}{l}
      {\color{blue} $\circ$} nodes in $\nodesgin$\quad\quad
      \\      
      {\color{black} $\bullet$} nodes in $\nodesgex$
      \\      
      {\color{red} $\times$} nodes in $\nodesgou$\quad\quad
    \end{tabular}
  \end{subfigure}
  \caption{Finite Element spaces.}
  \label{fig:def-spaces}
\end{figure}

%We consider the following constraints.

{\FV In order to define $\fespag$ we observe that, in nodal Lagrangian \ac{fe} spaces, there is a one-to-one map between \acp{dof} and nodes (points) {\FV of the \ac{fe} mesh} (for vector spaces, the same is true for every component of the vector field). On the other hand, we can define the owner {\ac{vef}} of a node as the lowest-dimensional {\ac{vef}} that contains the node. Furthermore, we can construct a map that for every {\ac{vef}} $F$ such that $F \not\subset \Dom$,  gives a cell owner among all the cells that contain it. This map can be arbitrarily built. E.g., we can consider as cell owner the one in the smallest aggregate. As a result, we have a map between \acp{dof} and (active) cells. Every active cell belongs to an aggregate, which has its own root (interior) cell.  So, we  also have a map between \acp{dof} and interior cells. This map between $b\in \nodesgou$ and the corresponding interior cell is represented with $\cell(b)$ (see Fig. \ref{fig:def-outnodmap}).

\begin{figure}[ht!]
  \centering
  \begin{subfigure}{0.28\textwidth}
    \centering
    \includegraphics[width=0.9\textwidth]{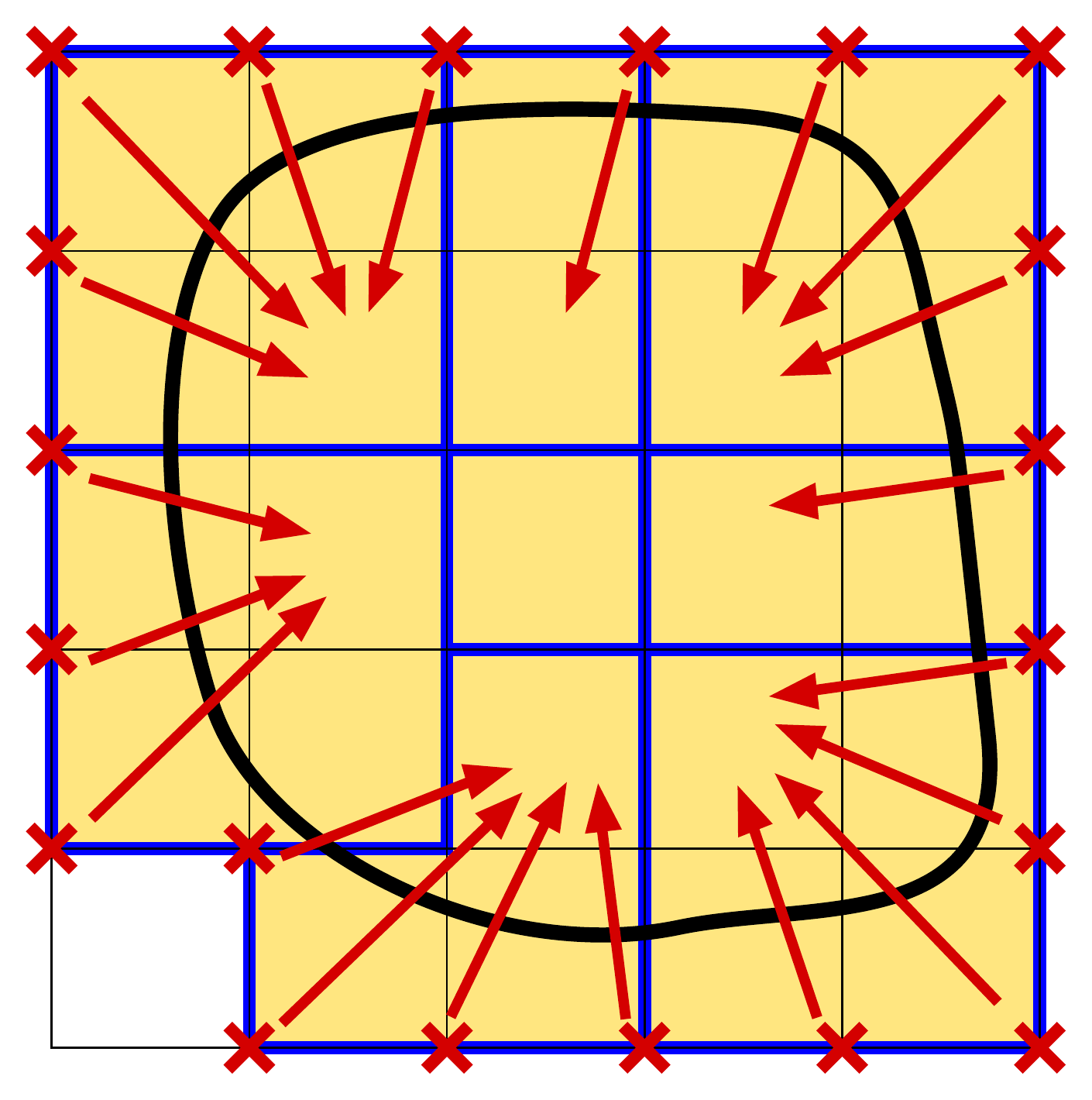}
  \end{subfigure}
  \begin{subfigure}{0.2\textwidth}
    \begin{tabular}{cl}
          \tikz{\draw[color=blue,fill=myellow, line width=2pt]  (0,0) rectangle (1.4em,1.4em);} & aggregate
          \\
      {\color{red} $\times$} & node in $\nodesgou$
      \\
      \includegraphics[width=2em]{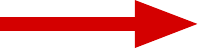} & node to cell map
    \end{tabular}
  \end{subfigure}
  \caption{Map from outer nodes to interior cells.}
  \label{fig:def-outnodmap}
\end{figure}

The space of global shape functions of $\fespin$ and $\fespex$ can be represented as $\{ \shpf{b} \, : \, b \in \nodesgin\}$ and $\{ \shpf{b} \, : \, b \in \nodesgex\}$, respectively. Functions in these \ac{fe} spaces are uniquely represented by their nodal values. We represent the nodal values of $\uh \in \fespin$ as $\uvin \in \mathbb{R}^{|\nodesgin|}$, whereas the nodal values of $\uh \in \fespex$ as $\uvex \in \mathbb{R}^{|\nodesgex|}$. Considering, without loss of generality, that the interior nodal values are labeled the same way for both \ac{fe} spaces, we have that $\uvex = [ \uvin, \uvou ]^T$, where $\uvou \in \mathbb{R}^{|\nodesgou|}$.

 \def\range{{\rm range}}

Now, we consider the following extension operator. Given $\uh \in \fespin$ and the corresponding nodal values $\uvin$, we compute the outer nodal values as follows:
\begin{equation}
\uvou_b = \sum_{a \in \nodes{\cell(b)}} \shpf{a}(\x_b) \uvin_a, \quad  \hbox{for} \ b \in \nodesgou.
\label{eq:def-constraints}
\end{equation}
{\FV That is, the value at an outer node $b\in\nodesgou$ is computed by extrapolating the nodal values of the  interior cell $K(b)$ associated with it.}
In compact form, we can write it as $\uvou = \cmat \uvin$, where $\cmat$ is the global matrix of constraints. We define the global extension matrix $\next: \mathbb{R}^{|\fespin|} \rightarrow \mathbb{R}^{|\fespex|}$ as $\next \uvin = [ \uvin, \cmat \uvin]^T$.  Let us also define the extension operator  $\mathcal{E}: \fespin \rightarrow \fespex$, such that, given $\uh \in \fespin$ represented by its nodal values $\uvin$, provides the \ac{fe} {\FV function} $\ext{\uh} \in \fespex$ with nodal values $\next{\uvin}$. We define the range of this operator as $\fespag \doteq \range(\ext{\fespin}) \subset \fespex$. This \ac{fe} space is called the \emph{aggregated} \ac{fe} space {\FV since the map $K(\cdot)$ between outer nodes and interior cells is defined using the aggregates in $\meshag$}. The motivation behind the construction of such space is to have a \ac{fe} space covering $\Domex$ (and thus $\Dom$) with optimal approximability properties and without the ill-conditioning problems of $\fespex$. 

As one can observe, the new space is defined only by interior nodal values, whereas the conflictive {\FV outer} nodes are \emph{eliminated} via the  constraints in \eqref{eq:def-constraints}. These constraints are cell-wise local. Thus, they can be readily applied at the assembly level in the cell loop, making its implementation very simple, even for non-adaptive codes that cannot deal with non-conforming meshes. 
%
% No crec que sigui veritat que la FE function dins d'um aggregate només depengui dels valors nodals del seu root.
We consider as basis for $\fespag$ the extension of the shape functions of $\fespin$, i.e., $\{ \ext{ \shpf{ a } }\}_{ a \in \nodesgin }$. The fact that it is a basis for $\fespag$ is straightforward, due to the fact that the extension operator is linear. The extension of a shape function is easily computed as follows:
$$
\ext{ \shpf{ a } } = \shpf{a} + \sum_{b \in \mathcal{C}(a)} \cmat_{ba} \shpf{b}, \quad \hbox{for } \, a \in \nodesgin,
$$
where $\mathcal{C}(a)$ represents the set of outer nodes in $\nodesgou$ that are constrained by $a$.
\begin{remark}
We note that one could consider an alternative aggregated space,
$$
\fespagnew = \{ v \in {\mathcal{C}^0}(\Omega) \, : \, v|_A \in
V(A), \, \hbox{for any} \, A \in \meshag \},
$$
where $V(A)$ denotes the space of $q$ order Lagrangian polynomials on n-simplices or n-cubes. It is obvious to check that in fact $\fespagnew \subset \fespag$, but it is not possible to implement the inter-element continuity for this space using standard \ac{fe} techniques. On the other hand, the \ac{fe} space $\fespag$  has the same size as the interior problem and the implementation in existing \ac{fe} codes requires minimal modifications. Furthermore, it is also easy to check that the two approaches coincide for \ac{dg} formulations, where all \acp{dof} belong to the cells. In fact, a \ac{dg} method with $\fespagnew$ has been proposed in \cite{kummer_extended_2013}.
\end{remark}

%\task{Define the extended fe space}

%\task{Define the extension operator}

%\task{Define the constraint matrix}

%\task{Define the nodal values extension operator}

%\task{Define the constrained extended fe space (aggregated fe space)}

%\task{Remark about the piecewise polynomial space on aggregates, subspace and hard to implement}

\section{Approximation of elliptic problems} \label{sec_app_ell}

For the sake of simplicity, we consider the Poisson equation with constant physical diffusion as a model problem, even though the proposed ideas apply to any elliptic problem with $H^1$-stability, e.g., the linear elasticity problem and heterogeneous problems. The Poisson equation with Dirichlet and Neumann boundary conditions reads as (after scaling with the diffusion term): find $u \in H^1(\Omega)$ such that 
\begin{equation}
%\left\lbrace
%\begin{array}{rl}
-\Delta u = f  \quad \text{in } \ \Omega, \qquad 
u=g^\mathrm{D}  \quad\text{on } \ \Gamma_\mathrm{D}, \qquad
\grad u \cdot \normal = g^\mathrm{N}  \quad \text{on } \ \Gamma_\mathrm{N},
%\end{array}
%\right.
\label{eq:PoissonEq}
\end{equation}
where $(\Gamma_\mathrm{D},\Gamma_\mathrm{N})$ is a partition of the domain boundary (the Dirichlet and Neumann boundaries, respectively), $f\in H^{-1}(\Omega)$, $g^\mathrm{D}\in H^{1/2}(\Gamma_\mathrm{D})$, and $g^\mathrm{N}\in H^{-1/2}(\Gamma_\mathrm{N})$. 

{\FV
For the space discretization, we consider $H^1$-conforming \ac{fe} spaces
on the conforming mesh $\meshact$ that are not necessary aligned with the the physical boundary $\partial\Dom$. For simplicity, we assume that, for any cut cell $\cell \in \meshex$, either $\cell \cap \Gamma \subset \Gamma_{\rm D}$ or $\cell \cap \Gamma \subset \Gamma_{\rm N}$. We consider both the usual \ac{fe} space  $\fespex$ as well as the new aggregated space $\fespag$ in order to compare their properties. We will simply use $\fespst$ when it is not necessary to distinguish between   $\fespex$ and $\fespag$.
}

%(The \ac{dg} case will not
%be considered in this work, but we could also develop preconditioned
%iterative linear solvers for unfitted DG formulations by combining the
%ideas herein with the BDDC preconditioner for DG formulations in
%\cite{dryja_bddc_2007}.)
 For unfitted grids, it is not clear to
include Dirichlet conditions in the approximation space in a strong
manner. Thus, we consider Nitsche's method
\cite{becker_mesh_2002,nitsche_uber_1971} to impose Dirichlet boundary
conditions weakly on $\Gamma_D$. It provides a consistent numerical
scheme with optimal converge rates (also for high-order elements) that
is commonly used in the embedded boundary community
\cite{schillinger_finite_2014}. We define the \ac{fe}-wise operators:
\begin{align}
  & \A_\cell(u,v) \doteq  \int_{\cell \cap \Omega} \gradient u \cdot \gradient v \mathrm{\ d}V + \int_{\Gamma_\mathrm{D} \cap \cell}\left( \tau_\cell u v  - v \left(\normal\cdot \gradient u\right) -  u \left(\normal\cdot \gradient v\right) \right)  \mathrm{\ d}{S},   \\
& \ell_\cell(v) \doteq {\int_{\Gamma_\mathrm{D} \cap \cell}\left( \tau_\cell  vg^\mathrm{D}  -  \left(\normal\cdot \gradient v\right)g^\mathrm{D} \right)  \mathrm{\ d}S}, & 
\end{align}
    {defined for a generic cell $\cell\in\meshact$}. Vector
    $\normal$ denotes the outwards normal to $\partial\Omega$. The
    bilinear form $\A_\cell(\cdot,\cdot)$ includes the usual form
    resulting from the integration by parts of \eqref{eq:PoissonEq} and
    the additional term associated with the weak imposition of
    Dirichlet boundary conditions with Nitsche's method. The
    right-hand side operator $\ell_\cell(\cdot)$ includes additional terms
    related to Nitsche's method. 
    The coefficient $\tau_\cell>0$ is a mesh-dependent parameter that has
    to be large enough to ensure the coercivity of
    $\A_\cell(\cdot,\cdot)$.

The global \ac{fe} operator $\Ac: \Vc \rightarrow \Vc'$  and right-hand side term $\ell \in \Vc' $ are stated as the sum of the element contributions, i.e.,
\begin{align}\label{eq:bil-form}
\Ac(u,v) \doteq \sum_{\cell \in \meshact} \A_\cell(u,v),  \quad
\ell{ (v)} \doteq \sum_{\cell \in \meshact} \ell_\cell{ (v)},
\qquad  \hbox{for } u, v \in \Vc. 
\end{align}

{\FV
We will make abuse of notation, using
    the same symbol for a bilinear form, e.g., $\A
    : \Vc \rightarrow \Vc'$, and its corresponding linear operator, i.e., $\langle \A u , v \rangle \doteq \A( u
    , v)$. 
    %Similarly for linear forms $f\in\Vc'$, $\langle f , v \rangle = f( v)$.
    Furthermore, we define $b: \Vc' \rightarrow \Vc$ as $b(v) \doteq f(v) + g^{\rm N}(v) + \ell(v)$, for $v \in \Vc$. With this, the global problem can be stated as: find $\uh \in \fespag$ such that
    \begin{equation} \label{eq:nummet} \Ac (\uh, \vh)  = b(\vh), \qquad \hbox{for any} \, \vh \in \fespag. \end{equation}
    By definition, this problem can analogously be stated as: find $\uh  \in \fespin$ such that $\Ac ( \ext{\uh}, \ext{\vh})  = b(\ext{\vh})$ for any $\vh \in \fespag$. 
    After the definition of the \ac{fe} basis (of shape functions) that spans $\fespag$, or alternatively the extension operator $\ext{\cdot}$,  the previous problem leads to a linear system to be solved.

%We compute the minimal values of $\beta$ required to  ensure coercivity with local element-wise eigenvalue problems, as proposed in \cite{de_prenter_condition_2017}. 
A sufficient (even though not necessary) condition for $\Ac$ to be coercive is to enforce the element-wise constant coefficient $\tau_\cell$ to satisfy
\begin{equation}
\tau_\cell \geq C_\cell \doteq \sup_{v\in V(\cell)} \dfrac{ \mathcal{B}_\cell (v,v)  }{ \mathcal{D}_\cell(v,v) },
\label{eq:coer-cond}
\end{equation}
for all the mesh elements $\cell\in\meshex$ intersecting the boundary $\Gamma_\mathrm{D}$. In the previous formula,  $\mathcal{D}_\cell(\cdot,\cdot)$ and  $\mathcal{B}_\cell(\cdot,\cdot)$ are the forms defined as
\begin{equation}
\mathcal{D}_\cell(u,v)\doteq \int_{\cell \cap \Omega} \gradient u \cdot \gradient v \mathrm{\ d}V,\quad \text{ and }\quad
\mathcal{B}_\cell (u,v) \doteq  \int_{\Gamma_\mathrm{D} \cap \cell} \left(\normal\cdot \gradient u\right) \left(\normal\cdot \gradient v\right) \mathrm{\ d}S.
\end{equation}

Since $\Vc$ is finite dimensional, and $\mathcal{D}_\cell(\cdot,\cdot)$ and  $\mathcal{B}_\cell(\cdot,\cdot)$ are symmetric and bilinear forms, the value  $C_\cell$ (i.e., the minimum admissible coefficient $\tau_\cell$)
can be computed numerically as $C_\cell\doteq\tilde\lambda_\mathrm{max}$, being $\tilde\lambda_\mathrm{max}$  the largest eigenvalue of the generalized eigenvalue problem  (see \cite{de_prenter_condition_2017} for details): find $u_\cell\in \fespst|_\cell$ and $\tilde\lambda\in\mathbb{R}$ such that
\begin{equation}
\mathcal{B}_\cell(u_\cell,v_\cell) = \tilde\lambda \mathcal{D}_\cell(u_\cell,v_\cell) \text{ for all } v_\cell\in \fespst|_\cell.
\label{eq:local-eigenvals}
\end{equation}  

{\FV
For standard \acp{fe} for body-fitted meshes, it is enough to compute coefficient $\tau_\cell$ as $\tau_\cell=\beta/h$ to satisfy condition \eqref{eq:coer-cond}, where $\beta$ is a sufficiently large (mesh independent) positive constant (see, e.g., \cite{arnold_unified_2002}).  However, for standard unfitted \ac{fe} methods using the usual space $\fespex$ without any additional stabilization, coefficient $\tau_\cell$ cannot be computed \emph{a priori}; in fact, the minimum cell-wise value that assures coercivity is not bounded above. In this case, a value for $\tau_\cell$ ensuring coercivity has to be computed for each particular setup using 
the cell-wise eigenvalue problem~\eqref{eq:local-eigenvals}. The introduction of the new space $\fespag$ solves this problem and  $\tau_\cell$ is bounded again in terms of the element size as expected in the body-fitted case (see Section \ref{sec:coe_nit} for more details). In this case, we have taken $\tau_\cell=100/\hc$ in the numerical experiments below.

}  

%We provide more details about the choice of the constant $\beta_\cell$ in \TBD{}. %In this work, we chose $\beta_\cell\doteq 2\tilde\lambda_\mathrm{max}$ in all the examples in order to be on the safe side, as suggested in~\cite{de_prenter_condition_2017}.
%Both operators have the same kernel, the space of constants. Thus, the generalized eigenvalue problem \eqref{eq:local-eigenvals} is well-posed (see~\cite{de_prenter_condition_2017} for more details). 

The linear system matrix that arises from \Eq{nummet} can be defined as 
\begin{equation}
\stifmatrix_{ab} \doteq \Ac( \ext{\shpf{a}} , \ext{\shpf{b}} ), \qquad \hbox{for } \, a, \, b \in \nodesgin.
\label{eq:stifmatdef} 
\end{equation}
The mass matrix related to the aggregated \ac{fe} space $\fespag$ is analogously defined as 
\begin{equation}
\massmatrix_{ab} \doteq \int_{\Omega} \ext{\shpf{a}} \ext{\shpf{b}}, \qquad \hbox{for } \, a, \, b \in \nodesgin.
\label{eq:massmatdef} 
\end{equation}
It is well known that the usual \ac{fe} space $\fespex$ is associated with conditioning problems due to cut cells. The condition number of the discrete system without the aggregation, i.e., considering $\fespex$ instead of $\fespag$ in \Eq{nummet}, scales as  
\begin{equation}
\condnum{\stifmatrix} \sim  \min_{\cell \in \meshex} \eta_\cell^{-(2q+1-2/d)},
\label{eq:finite-cell-estimate}
\end{equation}
where $\condnum{\stifmatrix}$ is the 2-norm condition number of $\stifmatrix$ (see \cite{de_prenter_condition_2017}  for details). Thus, arbitrarily high condition numbers are expected in practice since the position of the interface  cannot be controlled and the value  $\eta_\cell$ can be arbitrarily close to zero. This problem is solved if the new aggregated space $\fespag$ is used instead of $\fespex$ (cf. Corollary~\ref{cr:cn}). 
}

\section{Numerical analysis} \label{sec:sta_ana}

In this section, we analyze the well-posedness of the agregated unfitted \ac{fe} method \Eq{nummet}, the condition number of the arising linear system, and \emph{a priori} error estimates. As commented above, we assume that the background mesh is \emph{quasi-uniform}. Therefore, the number of neighboring cells of a given cell is bounded above by a constant $n_{\rm cell}$ independently of $\h$. In a mesh refinement analysis, we also assume that the coarser mesh level-set function already represents the domain boundary.

In the following analysis, \emph{all constants being used are independent of $\h$ and the location of the cuts in cells}, i.e., $\eta_\cell$. They may also depend on the threshold $\eta_0$ in the aggregation algorithm if considered; we have considered $\eta_0 = 1$ for simplicity. The constants can depend on the shape/size of $\Omega$ and $\Gamma_{\rm D}$, the order of the \ac{fe} space, and the maximum aggregation distance $\gamma_{\rm max}$, which are assumed to be fixed in this work. In turn, due to Lemma \ref{lemma:aggrsize}, the maximum size of an aggregate is bounded by a constant times $\h$. As a result, the following results are robust with respect to the so-called \emph{small cut cell problem}. When we have that $A \leq c B$ for a positive constant $c$, we may use the notation $A \lesssim B$; analogously for $\gtrsim$.

For the analysis below, we need to introduce some extra notation. Given a function $\uh \in \fespin$ (or $\fespex$), the nodal vector $\uv$ will be used without any superscript, as soon as it is clear from the context. For a given cell $\cell$, the cell-wise coordinate vector is represented with $\ucv$. On the other hand, given a \ac{fe} function $\uh\in\fespin$, for every interior cell $\cell \in \meshin$,
let us define define the cell-wise extension operator $\next_\aggr{\uv} = [\ucv, \cmatcref \ucv]$, 
where  $  \cmatcref $ is the cell-wise constraint matrix, whose entries can be computed in the reference space (see \Eq{def-constraints}), such that $\cmat \uv \cdot \cmat \uv = \sum_{\cell \in \meshin} \cmatcref \ucv \cdot \cmatcref \ucv $.  We denote with $\| \cdot \|_2$ the Euclidean norm of a vector and the induced matrix norm. Standard notation is used to define Sobolev spaces (see, e.g., \cite{agmon1965lectures}). Given a Sobolev space $X$, its corresponding norm is represented with $\| \cdot \|_X$.

%\task{Stability of the extension operator}

\subsection{Stability of the coordinate vector extension matrix}\label{sta_ext}

We start the analysis of the scheme by proving bounds for the norm of the global and cell-wise coordinate vector extension matrix.  
Therefore, their norms can be bounded independently of the cut location and the size of the aggregate.
\begin{lemma}\label{lm:extmatstab}
The cell-wise and global coordinate vector extension matrices hold the following bounds:
$$ 1 \leq  \| \next_\aggr \|_2^2 \leq 1 + \| \cmatcref \|_2^2, \quad \hbox{for every } \, \aggr \in \meshag, \quad \hbox{and} \quad 
1 \leq \| \next \|_2^2 \leq 1 + \| \cmat \|_2^2 \leq \cev,
$$
for a positive constant $\cev$. 
% It leads to the following bound for the global coordinate vector extension matrix
% \begin{align}\label{eq:next1}
% 1 \leq {\next{\uv}} \cdot {\next{\uv}}  \leq (1 + \|  \cmat \|^2)   \uv \cdot \uv \leq  \cev  \uv \cdot \uv,
% \end{align}
%where the constraint matrix norm is bounded as
%$$
%\|  \cmat \|^2 \leq n_{\rm cell} (1 + \sup_{\aggr \in \meshag} \| \cmatcref \|^2) \doteq \cev.
%$$
\end{lemma}
\begin{proof}
  Using the definition of the extension operator in Section \ref{sec:agg_unf}, we have that $\| {\next{\uv}} \|_2^2 =  \| \uv \|_2^2 + \| \cmat \uv \|_2^2$. 
We proceed analogously for the cell-wise result, to get $\| {\next_\aggr{\uv}} \|_2^2 =  \| \uv \|_2^2 + \| \cmatcref \uv \|_2^2$. It proves the first result. On the other hand, we have,
\begin{align}
  \| \cmat \uv \|_2^2  & =  \sum_{\aggr \in \meshag} \| \cmatcref \ucv \|_2^2 \leq  \sum_{\aggr \in \meshag} \| \cmatcref \|_2^2 \| \ucv \|_2^2 \leq  n_{\rm cell}  \sup_{\aggr \in \meshag} \| \cmatcref \|_2^2 \| \uv \|_2^2,
\end{align}
where we have used the fact that the constraint matrix is aggregate-wise and that the maximum number of cell neighbors of a vertex/edge/face is bounded above by a constant $n_{\rm cell}$. The value $\sup_{\aggr \in \meshag} \| \cmatcref \|_2^2$ (or an upper bound) can explicitly be computed prior to the numerical integration and its entries are independent of the aggregate cut and the geometrical mapping, i.e., $\h$. In fact, given a polynomial order and $\gamma_{\rm max}$, one can precompute the maximum value of $\|\cmatcref \|_2^2$ among all possible aggregate configurations and explicitly obtain an upper bound $\cev$ of the global extension matrix norm. It proves the lemma.
\end{proof}

\subsection{Mass matrix condition number} \label{mas_mat}
In order to provide a bound for the condition number of the mass matrix, we rely on the maximum and minimum eigenvalues of the local mass matrix in the reference cell $\cellref$:
\begin{align}\label{eq:eigref_mass}
  \eigmin \| \ucv \|_2^2 \leq \| \uh \|_{L^2(\cellref)}^2 \leq \eigmax \| \ucv \|_2^2, \qquad \hbox{for } \, \uh \in V(\cell).
\end{align}
The values of $\eigmin$ and $\eigmax$ only depend on the order of the \ac{fe} space and can be computed for different orders on n-cubes or n-simplices (see \cite{elman_finite_2005}). Using typical scaling arguments, one has the following bound for the local mass matrix of the physical cell:
\begin{align}\label{eq:eig_mass}
\eigmin \hc^d \| \ucv \|_2^2 \leq  \|  \uh \|^2_{L^2(\cell)} \leq \eigmax \hc^d \| \ucv \|_2^2.
\end{align}
In the next lemma, we prove the equivalence between the $L^2(\Domex)$ norm and the interior \ac{dof} Euclidean norm, for functions in $\fespag$.
\begin{lemma}\label{lm:l2nodal}
The following bounds hold:
\begin{align}\label{eq:l2stab1}
h^d  \| \uv \|_2^2 \lesssim  \| \ext{\uh} \|^2_{L^2(\Domex)} \lesssim h^d  \| \uv \|_2^2, \quad \hbox{for any } \, \uh \in \fespin.
\end{align}
\end{lemma}
\begin{proof}   
  By definition, every function in $\fespag$ can be expressed as $\ext{\uh}$ for some $\uh \in \fespin$. Using \Eq{eig_mass}, the fact that $\Domin \subset \Dom$, and the quasi-uniformity of the background mesh, we obtain the lower bound in \Eq{l2stab1} as follows:
$$
\|  \ext{\uh} \|_{L^2(\Domex)}^2 \geq \|  \uh \|^2_{L^2(\Domin)} = \sum_{K \in \meshin} \|  \uh \|^2_{L^2(\cell)} \geq \sum_{K \in \meshin} \hc^d \eigmin \| \ucv \|_2^2 \gtrsim \h^d \|  \uv \|_2^2.
$$
On the other hand, using $\Dom \subset \Domex$, Lemma \ref{lm:extmatstab}, \Eq{eig_mass}, and the fact that the number of surrounding cells of a node is bounded above by a positive constant, we get:
\begin{align}
\|  \ext{\uh} \|_{L^2(\Domex)}^2 = \sum_{K \in \meshex}  \|  \ext{\uh} \|^2_{L^2(\cell)}   \leq \sum_{K \in \meshex}  \h^d_\cell \eigmax \| \ucv \|_2^2  \lesssim  \h^d  \| \next{\uv} \|_2^2   \lesssim  \h^d  \| \uv  \|_2^2.
\end{align}
It proves the lemma.
\end{proof}

The upper and lower bounds in \Eq{l2stab1} lead to the continuity of the extension operator and a bound for the condition number of the mass matrix of the aggregated \ac{fe} space.

\begin{corollary}[Continuity of the  extension operator]\label{lm:stab_ext_op}
  The extension operator satisfies the following bound:
  \begin{align}
    \| \ext{\uh} \|^2_{L^2(\Domex)} \lesssim  \| \uh \|^2_{L^2(\Domin)}, \qquad \hbox{for any } \, \uh \in \fespin.
    \end{align}
  \end{corollary}
 %  \begin{proof}
%   In order to prove this bound, we combine the upper and lower bounds in \Eq{l2stab1} as follows:
%   $$
%   \int_\Dom \ext{\uh}^2 \leq n_{\rm cell} h^d \cev \eigmax \uv \cdot \uv \leq \frac{n_{\rm cell} \cev\eigmax}{\eigmin}  \int_\Domin \uh^2. 
%   $$
% \end{proof}

  \begin{corollary}[Mass matrix condition number]
   The mass matrix $\massmatrix$ in \Eq{massmatdef}, related to the aggregated \ac{fe} space $\fespag$, is bounded  by $\condnum{\massmatrix} \leq C$, for a positive constant $C$.
%as follows:
%$$
%\condnum{\massmatrix} \leq \frac{n_{\rm cell} (1 + \| \cmat \|^2) \eigmax}{\eigmin}.
%$$
\end{corollary}

\subsection{Inverse inequality} 
In order to prove the condition number bound for the system matrix arising from \Eq{nummet}, we need to prove first an extended inverse inequality. We rely on the fact that an inverse inequality holds for the \ac{fe} space $\fespex$, i.e., 
\begin{equation}\label{eq:inv_act}
\|  \grad \uh \|_{L^2(\Domex)} \lesssim h^{-1} \|  \uh \|_{L^2(\Domex)}, \qquad \hbox{for any } \, \uh \in \fespex. 
\end{equation}
This standard result for conforming meshes can be found, e.g., in \cite[p. 111]{brenner_mathematical_2010}.
\begin{lemma}[Inverse inequality]
\label{lm:inv_ineq}
  The following inverse inequality holds:
$$
\|  \grad \ext{\uh} \|_{L^2(\Domex)} \lesssim h^{-1}  \|  \uh \|_{L^2(\Domin)}, \qquad \hbox{for any } \uh \in \fespin.
$$
\end{lemma} 
\begin{proof}
Using the fact that $\Domin \subseteq \Dom \subseteq \Domex$,  $\ext{\uh} \in \fespex$, the standard inverse inequality \Eq{inv_act}, and the stability of the extension operator in  Lemma \ref{lm:stab_ext_op}, we get:  
\begin{align}
  \|\grad \ext{\uh}\|_{L^2(\Domex)} \lesssim  h^{-1} \| \ext{\uh} \|_{L^2(\Domex)} \lesssim   h^{-1} \| \uh \|_{L^2(\Domin)}.
\end{align}
It proves the lemma.
\end{proof}

\subsection{Coercivity and Nitsche's coefficient}\label{sec:coe_nit}

\def\xiv{{\underline{\boldsymbol{\xi}}}}
\def\dmat{{\mathbf{D}}}
\def\xih{{{\boldsymbol{\xi}_h}}}
\def\cellcut{\cell_{\rm cut}}

In this section, we consider a trace inequality that is needed to prove the coercivity of the bilinear form in \Eq{bil-form}. Given a  cell $\cell \in \meshact$, let us consider the set of \emph{constraining} interior cells $\cell_1, \ldots, \cell_{m_\cell}$,  $m_\cell \geq 1$, i.e., the interior cells that constraint at least one \ac{dof} of the cut cell. Let us also define $\cellcut \doteq \cell \cap \Dom$ and $\Omega_{\cell} \doteq \cellcut \cup  \bigcup_{i=1}^{m_\cell} K_i \subset \Dom$.
\begin{lemma}
For any $\uh \in \fespag$ and $\cell \in \meshact$, the following bound holds
$$
\| \n\cdot \grad \uh \|_{L^2({\Gamma_{\rm D} \cap \cell})} \leq \cbou h_{\cell}^{-\frac{1}{2}} \| \grad \uh \|_{L^2(\Dom_{\cell})},
$$
for a positive constant $\cbou$.
\end{lemma}
\begin{proof}
For interior cells, the left-hand side is zero and the bound trivially holds. Let us consider a cut cell $\cell$. Let us also consider a \ac{fe} function $\uh \in \fespag$ and its gradient $\xih \doteq \grad \uh$. Assuming that all the cells have the same order, we have that $\xih$ belongs to the discontinuous Lagrangian \ac{fe} space of order $q-1$, and we represent the corresponding coordinate vector with $\xiv_\cell$.

First, we use the equivalence of norms in finite dimension and a scaling argument to get:
$$
\| \xih \|_{L^2(\Gamma_{\rm D} \cap \cell)}^2 \lesssim |\Gamma_{\rm D} \cap \cell| \| \xih \|^2_{L^\infty(\cell)}, 
$$
where the constant can only depend on the \ac{fe} space order. Analogously, we have $\| \xih \|_{L^{\infty}(\cellref)}^2 \lesssim \| \xiv_\cell\|_2^2$. Following the same ideas as above, $\xiv$ can be expressed as an extension of the corresponding nodal values of the $q-1$ order \ac{fe} spaces on top of the interior cells $\cell_i$, represented with $\xiv_{\cell_i}$; we represent this extension with the matrix $\dmat_\cell$, i.e., $\xiv_\cell = \dmat_{\cell} [ \xiv_1, \ldots, \xiv_m]^T$. Using an analogous reasoning as above for matrix $\cmat$, the norm of this matrix cannot depend on the cut or $h$. Thus, we have that $\| \xiv_\cell\|_2^2 \lesssim  \sum_{i=1}^{m_\cell} \| \xiv_{\cell_i} \|_2^2$. On the other hand, using again the equivalence of norms in finite dimension, we get $\| \xiv_{\cell_i} \|_2  \lesssim  \| \xih \|_{L^2(\cellref_i)}$. As a result, using typical scaling arguments, and using the fact that $|\cell| \lesssim |\cell_i| \lesssim | \cell|$ for constants independent of mesh size or order, we get:
$$\| \xih \|_{L^{\infty}(\cell)}^2 \lesssim | \cell |^{-1}  \sum_{i=1}^{m_\cell}\| \xih \|^2_{L^2(\cell_i)}.$$

Combining these results, we get:
\begin{align}
  \|\n\cdot \xih \|_{L^2(\Gamma \cap \aggr)}^2 & \leq |\Gamma \cap \cell| |\cell | \| \xih \|^2_{L^\infty(\cell)}  \lesssim   \hc^{-1}  \sum_{i=1}^{m_\cell}\| \xih \|^2_{L^2(\cell_i)},
  \end{align}
  where we have used the fact that $|\Gamma \cap \cell| |\cell |^{-1} \lesssim \hc^{-1}$ holds for a quasi-uniform mesh.
  It proves the lemma. 
\end{proof}

\subsection{Well-posedness of the unfitted \ac{fe} problem}
In this section, we prove coercivity and continuity of the bilinear form \Eq{bil-form}. First, we prove coercivity with respect to the following mesh dependent norm in $\fespag$:
$$
\opnormh{\uh}^2 \doteq \| \grad \uh \|_{L^2(\Dom)}^2 + \sum_{\cell \in \meshact} \beta_\cell \hc^{-1} \| \uh \|^2_{L^2(\Gamma_\mathrm{D} \cap \cell)}, \qquad \hbox{for } \, \uh \in \fespag,
$$ 
which is next proved to bound the $L^2(\Dom)$ norm.
\begin{theorem}\label{th:wp}
  The aggregated unfitted \ac{fe} problem in \Eq{nummet} satisfies the following bounds:
  \begin{itemize}
  \item[i)] Coercivity: 
\begin{equation}\label{eq:th-bounds}
  \Ac(\uh,\uh) \gtrsim  \opnormh{\uh}^2, \quad \hbox{for any } \uh \in \fespag,
\end{equation}
  \item [ii)] Continuity:
\begin{equation}\label{eq:th-bounds-2}
\Ac(\uh,\vh) \lesssim  \opnormh{\uh} \opnormh{\vh}, \quad \hbox{for } \uh, \, \vh \in \fespag,
\end{equation}
    \end{itemize} 
if $\beta_K > C$, for some positive constant $C$. In this case, there exists one and only one solution of \Eq{nummet}.
\end{theorem}
\begin{proof}
  For cut cells, we use
  \begin{align}
    2 \int_{\Gamma_{\rm D} \cap \cell} \uh \left(\normal\cdot \gradient \uh\right) & \leq  \alpha_\cell {\cbou} h_\cell^{-1}  \| \uh \|^2_{L^2({\Gamma_{\rm D} \cap \cell})} +  \alpha_\cell^{-1} \cbou^{-1} h_\cell  \|  \normal\cdot \gradient \uh \|^2_{L^2({\Gamma_{\rm D} \cap \cell})} \\ & \leq \alpha_\cell {\cbou} h_\cell^{-1} \| \uh \|^2_{L^2({\Gamma_{\rm D} \cap \cell})}   +\alpha_\cell^{-1}   \| \grad \uh \|^2_{L^2(\Dom_\cell)}.
\label{eq:stabA1-proof}
 \end{align}
  Using the fact that the mesh is quasi-uniform and that the number of neighboring cells and $\gamma_{\rm max}$  is bounded, one can take a value for $\alpha_\cell$ large enough (but uniform with respect to $\h$ and the cut location) such that:
  $$
  2 \int_{\Gamma_{\rm D}} \uh \left(\normal\cdot \gradient \uh\right) \leq   \sum_{\cell \in \meshact }\alpha_\cell\cbou \h_\cell^{-1}\| \uh \|_{L^2(\Gamma_{\rm D}\cap \cell)}^2 +  \frac{1}{2}\| \grad \uh \|_{L^2(\Dom)}^2.
$$
As a result, we get:
$$
\Ac(\uh,\uh) \geq  \frac{1}{2} \| \grad \uh \|_{L^2(\Dom)}^2 + \sum_{\cell \in \meshact } \left( \beta_\cell - \alpha_\cell  \cbou \right) \h_\cell^{-1} \| \uh \|_{L^2(\Gamma_{\rm D}\cap \cell)}^2.
$$
For, e.g., $\beta_\cell > 2 \alpha_\cell \cbou$, $\Ac(\uh,\uh)$ is a norm. By construction, this lower bound for $\beta_\cell$ is independent of the mesh size $\h$ and the intersection of $\Gamma$ and $\meshact$. It proves the coercivity property in \Eq{th-bounds}. Thus, the bilinear form is non-singular.  The continuity in \Eq{th-bounds-2} can readily be proved by repeated use of the Cauchy-Schwarz inequality and inequality \Eq{stabA1-proof}. Since the problem is finite-dimensional and the corresponding linear system matrix is non-singular, there exists one and only one solution of this problem. 
\end{proof}
%%%%%%%%%%%%%%
%%%%%%%%%%%%%%
%%%%%%%%%%%%%%

\def\Cstab{C_{\rm s}}

\begin{lemma}\label{lm:l2meshnorm}
If $\Dom$ has smoothing properties, the following bound holds:
$$
\| \uh \|_{L^2(\Dom)} \lesssim %\leq \Cstab(\Gamma_{\rm D},\Dom)
\opnormh{\uh}, \quad \hbox{for any } \, \uh \in \fespag.
$$
\end{lemma}
\begin{proof}
Let us consider $\uh \in \fespag$ and let $\psi \in H^1_0(\Dom)$ solve the problem $-\Delta \psi = \uh$ with the boundary conditions $\psi = 0$ on $\Gamma_{\rm D}$ and $\normal \cdot \grad \psi = 0$ on $\Gamma_{\rm N}$. Using the fact that the domain $\Omega$ has smoothing properties, it holds $\| \psi \|_{H^2(\Dom)} \lesssim  \| \uh \|_{L^2(\Dom)}$. We have, after integration by parts:
\begin{align}\label{eq:stab-eq1}
\| \uh \|_{L^2(\Dom)} = - \int_\Dom \uh \Delta \psi = \int_\Dom \grad \uh \cdot \grad \psi - \int_{\Gamma_{\rm D}} \uh \normal \cdot \grad \psi.
\end{align}
The first term in the right-hand side of \Eq{stab-eq1} is easily bounded using the Cauchy-Schwarz inequality:
$$
\int_\Dom \grad \uh \cdot \grad \psi \leq  \| \grad \uh \|_{L^2(\Dom)} \|  \grad \psi \|_{L^2(\Dom)} \lesssim \| \grad \uh \|_{L^2(\Dom)} \|  \uh \|_{L^2(\Dom)}. 
$$
On the other hand, the following trace inequality  holds $\| \normal \cdot \grad \psi \|_{L^2(\Gamma_{\rm D})}^2 \lesssim | \psi |_{H^2(\Omega)}^2$ for a constant that depends on the size of $\Gamma_{\rm D}$ (see \cite{agmon1965lectures}). Using the Cauchy-Schwarz inequality and the previous trace inequality, we readily get:
\begin{align}
  - \int_{\Gamma_{\rm D}} \uh \normal \cdot \grad \psi &\leq \left( \sum_{\cell \in \meshact} \hc^{-1} \| \uh \|_{L^2(\Gamma_{\rm D} \cap \cell)}^2 \right)^{\frac{1}{2}}
  \left( \sum_{\cell \in \meshact} \hc \| \normal \cdot \grad \psi \|_{L^2(\Gamma_{\rm D} \cap \cell)}^2 \right)^{\frac{1}{2}} \\
  & \leq \left( \sum_{\cell \in \meshact} \hc^{-1} \| \uh \|_{L^2(\Gamma_{\rm D} \cap \cell)}^2 \right)^{\frac{1}{2}} \| \normal \cdot \grad \psi \|_{L^2(\Gamma_{\rm D})} \\ & \lesssim  \left( \sum_{\cell \in \meshact} \hc^{-1} \| \uh \|_{L^2(\Gamma_{\rm D} \cap \cell)}^2 \right)^{\frac{1}{2}} \| \uh \|_{L^2(\Dom)}.
\end{align}
Combining these bounds, we prove the lemma.
\end{proof}

\def\cstif{C_{\rm A}}
\begin{corollary}[Stiffness matrix condition number]
  \label{cr:cn}
The condition number of the linear system matrix $\stifmatrix$ in \Eq{stifmatdef} 
is bounded  by $\condnum{\stifmatrix} \lesssim  \h^{-2}$. 
\end{corollary}
\begin{proof}
To prove the corollary, we have to bound $\uv \cdot \stifmatrix \uv = \Ac(\ext{\uh},\ext{\uh})$ above and below by $ \| \uv \|^2_2$ times some constant. The lower bound follows from the coercivity property in Th. \ref{th:wp}, Lemma \ref{lm:l2meshnorm}, the lower bounds in Lemmas \ref{lm:l2nodal} and \ref{lm:extmatstab}, which lead to $\Ac(\ext{\uh},\ext{\uh}) \gtrsim \|  \ext{\uh} \|^2_{L^2(\Dom)} \gtrsim  h^d  \| \uv \|^2_2.$ The upper bound is readily obtained from the continuity property in Lemma  \ref{th:wp} and the upper bound in Lemma \ref{lm:l2meshnorm}, i.e., $
\uv \cdot \stifmatrix \uv = \Ac(\ext{\uh},\ext{\uh}) \lesssim \opnormh{\ext{\uh}}^2$. Using scaling arguments and the equivalence of norms for finite-dimensional spaces, we get $\| \ext{\uh} \|^2_{L^2(\Gamma_\mathrm{D} \cap \cell)} \lesssim \hc^{d-1} \| \next \ucv \|_2^2$. Adding up for all cells, invoking the fact that the number of neighbour cells is bounded, and using the upper bound of the coordinate vector extension operator in \ref{lm:l2meshnorm}, we obtain:
\begin{align}\label{eq:aux1-sma}
\sum_{\cell \in \meshact} \beta_\cell \hc^{-1} \| \ext{\uh} \|^2_{L^2(\Gamma_\mathrm{D} \cap \cell )} \lesssim \h^{d-2} \| \uv \|_2^2.
\end{align}
Using the inverse inequality in Lemma \ref{lm:inv_ineq} and the upper bound in Lemma \ref{lm:l2nodal}, we obtain:
\begin{align}\label{eq:aux2-sma}
\| \grad \ext{\uh} \|_{L^2(\Dom)}^2 \lesssim  h^{-2} \| \ext{\uh} \|_{L^2(\Dom)}^2 \lesssim h^{d-2} \| \uv \|_2^2.
\end{align}
Combining \Eq{aux1-sma}-\Eq{aux2-sma}, we get $\uv \cdot \stifmatrix \uv \leq c h^{d-2} \| \uv \|_2^2$. It proves the corollary.
\end{proof}
%
%as follows:
%$$
%\condnum{\massmatrix} \leq \frac{n_{\rm cell} (1 + \| \cmat \|^2) \eigmax}{\eigmin}.
%$$

%\section{Condition number estimates} \label{sec:con_num}

\subsection{Error estimates} \label{sec:err_est}

\def\int#1{\mathcal{I}_h(#1)}
\def\nodalval#1{{\underline{\boldsymbol{\sigma}}(#1)}}

In this section, we get \emph{a priori} error estimates for the aggregated \ac{fe} scheme \Eq{nummet}. In order to do that, we prove first approximability properties of the corresponding spaces.

\begin{lemma}\label{lm:approx}
  Let us consider an aggregated \ac{fe} space of order $q$, $m \leq q$, $1 \leq s \leq m \leq q+1$, $1 \leq p \leq \infty$, and  $m > \frac{d}{p}$. Given a function $u \in H^m(\Omega)$, it holds:
  $$
\inf_{\uh \in \fespag} \| u - \uh \|_{W_p^s(\Dom)} \lesssim \h^{m-s} | u |_{W_p^m(\Dom)}.
  $$
\end{lemma}

\begin{proof}

Under the assumptions of the lemma, we have that the following embedding $W_p^m(\Dom) \subset \mathcal{C}^0(\bar\Dom)$ is continuous (see, e.g.,  \cite[p. 486]{ern_theory_2004}).  Thus, given a function $u \in W_p^m(\Dom)$, let represent with $\nodalval{u}$ the vector of nodal values in $\Domin$, i.e., $\nodalval{u}_a = \uh(\x^a)$ for $a \in \nodes{\Domin}$. We define the interpolation operator $\int{u} \doteq \sum_{a \in \nodes{\Dom}} \shpf{a} [\next{\nodalval{u}}]_a $.
  
Given a cut cell $K \in \meshact$, the fact that its \acp{dof} values only depend on interior \acp{dof} in $\bar\Omega_\cell$, and since each shape function $\shpf{a}$ belongs to $W_\infty^m(K) \subseteq W_p^m(K)$, it follows from the upper bound of the norm of the nodal extension operator in Lemma \ref{lm:extmatstab} that  $\| \int{u} \|_{W_p^m(\cell)} \leq C \| u \|_{\mathcal{C}^0(\bar\Omega_\cell)}$ (see also \cite[Lemma 4.4.1]{brenner_mathematical_2010}). On the other hand, we consider an arbitrary function $\pi(u) \in W_p^m(\Dom)$ such that $\pi(u)|_\cell \in \mathcal{P}_q(\Dom_{\cell})$. We note that, by construction, $\pi(u)|_\cell = \int{\pi(u)}|_\cell$. Thus, we have:
%Now, we have, for $0 \leq s \leq m$:
\begin{align}
 \| u - \int{u} \|_{W_p^m(\cell)}  & \leq 
  \| u - \pi(u)  \|_{W_p^m(\cell)}  +  \| \int{\pi(u) - u} \|_{W_p^m(\cell)} \\
& \lesssim
  \| u - \pi(u)  \|_{W_p^m(\cell)}  +   \| \pi(u) - {u} \|_{\mathcal{C}^0(\bar\Dom_\cell)} \lesssim \| u - \pi(u)  \|_{W_p^m(\Dom_\cell)},
\end{align}
where we have used in the last inequality the previous continuous embedding. Since $\Dom_\cell$ is an open bounded domain with Lipschitz boundary by definition, one can use the Deny-Lions lemma (see, e.g., \cite{ern_theory_2004}). As a result, the $\pi(u)$ that minimizes the right-hand side holds:
\begin{align}
\| u - \int{u} \|_{W_p^m(\Dom_\cell)} \lesssim | u |_{W_p^m(\Dom_\cell)}.
\end{align}
Using standard scaling arguments, we prove the lemma.
\end{proof}

\begin{theorem}  
If $\Dom$ has smoothing properties and the solution $u$ of the continuous problem \Eq{PoissonEq} belongs to  $W_p^m(\Dom)$ for $\frac{d}{p} < m \leq q$, the solution $\uh \in \fespag$ of \Eq{nummet}  satisfies the following \emph{a priori} error estimate:
$$
\| u - \uh \|_{H^1(\Dom)} \leq \h^{m-1} | u |_{H^m(\Dom)}.
$$
\end{theorem}

\def\wh{w_h}
\def\vh{v_h}

\begin{proof}
  Combining the consistency of the numerical method, i.e., $\Ac(u,\vh) = \ell(\vh)$, and the continuity  and coercivity of the bilinear form in Th. \ref{th:wp}, we readily get, using standard \ac{fe} analysis arguments: 
  $$
  \opnormh{\wh-\uh}^2 \lesssim \Ac(\wh-\uh,\wh-\uh) = \Ac(\wh - u, \wh - \uh) \lesssim \opnormh{\wh-u}
\opnormh{\wh-\uh}, 
$$
for any $\wh \in \fespag$. On the other hand, we use the trace inequality (see \cite{arnold_interior_1982})
$$
\| \psi \|^2_{L^2(\partial T)} \lesssim  
  |\partial T|^{-1} \| \psi \|^2_{L^2(T)} + |\partial T| \| \psi \|^2_{H^1(T)}, \qquad \hbox{for any } \, \psi \in H^1(T).
$$
Using this trace inequality, we get:
$$ \hc^{-1} \| u- \wh \|^2_{L^2(\Gamma_\mathrm{D} \cap \cell)} \leq
\hc^{-1} \| u- \wh \|^2_{L^2(\partial \Dom_\cell)} \lesssim
\hc^{-2} \| u- \wh \|^2_{L^2( \Dom_\cell)}  + \| u- \wh \|^2_{H^1( \Dom_\cell)}. 
$$
Combining the previous bound with the approximability property in Lemma \ref{lm:approx}, we readily get
$$
\opnormh{\wh-u} \lesssim \h^{m-1} | u - \wh |_{H^m(\Dom)}.
$$
It proves the theorem.
\end{proof}

\section{Numerical experiments} \label{sec:num_exp}

\subsection{Setup} The numerical examples below consider as a model problem the Poisson equation with non-homogeneous Dirichlet boundary 
conditions. The value of the  source term and the Dirichlet function are defined such that the PDE has the following manufactured exact solution:
\begin{align}
& u(x,y,z) = \sin\left(4\pi\left(\left(x-2.3\right)^2 + y^2 + z^2\right)^{1/2}\right),\\ &  (x,y)\in\Omega\subset\mathbb{R}^2, z=0 \text{ in 2D}, \quad (x,y,z)\in\Omega\subset\mathbb{R}^3 \text{ in 3D}.
\label{eq:manuf-exact-sol}
\end{align}
We consider two different geometries, a 2D circle and a 3D complex domain with the shape of a popcorn flake (see Fig. \ref{fig:experimental-setup}). These geometries are often used in the literature to study the performance of unfitted \ac{fe} methods 
(see, e.g., \cite{burman_cutfem:_2015}, where the definition of the popcorn flake is found). In all cases, we use the cuboid $[0,1]^d$, $d=2,3$ as the bounding box on top of which the 
background Cartesian grid is created. For the sake of illustration, Fig. \ref{fig:experimental-setup} displays both the considered geometries, numerical solution and bounding box. 
\begin{figure}[ht!]
 \centering
  \begin{subfigure}{0.4\textwidth}
    \includegraphics[width=0.95\textwidth]{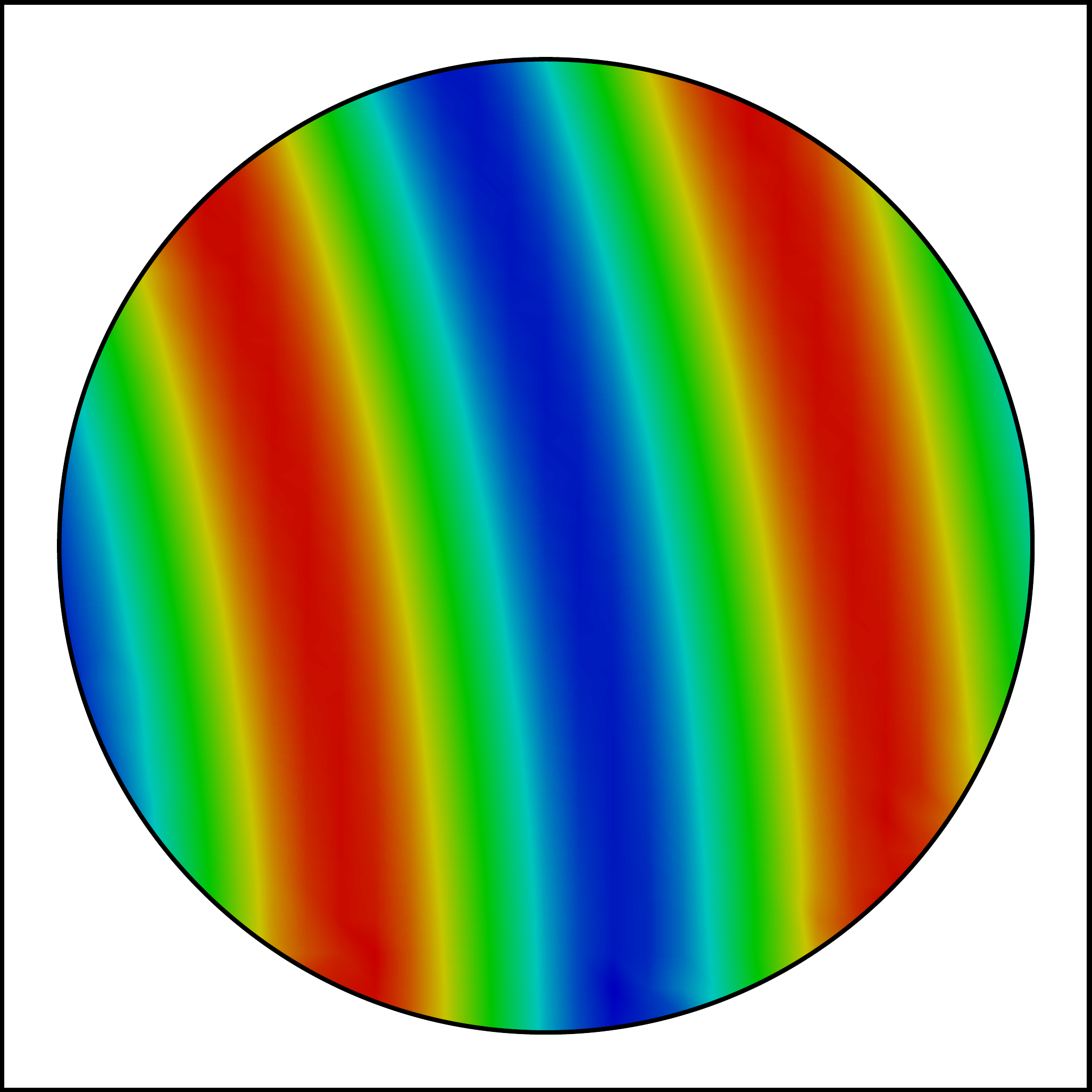}
    \caption{2D case}
  \end{subfigure}
  \begin{subfigure}{0.4\textwidth}
    \includegraphics[width=0.99\textwidth]{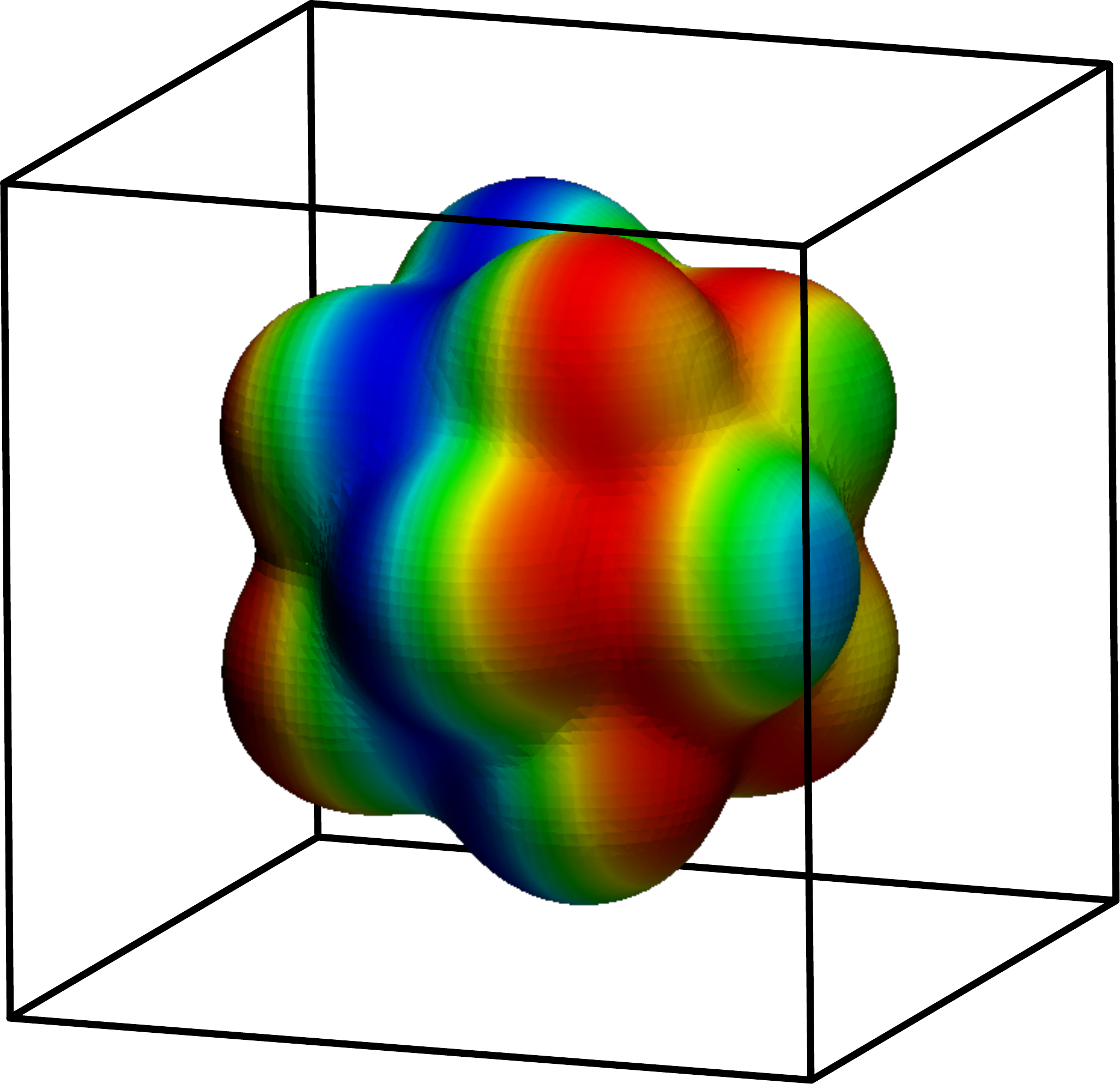}
    \caption{3D case}
  \end{subfigure}
  \begin{subfigure}{0.05\textwidth}
    \includegraphics[width=0.99\textwidth]{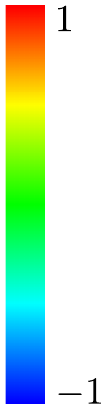}
  \end{subfigure}
  \caption{ View of the problem geometries, numerical solution and bounding box. }
  \label{fig:experimental-setup}
\end{figure}

The main goal of the following tests is to evaluate the (positive) effect of using the aggregation-based \ac{fe} space $\fespag$ instead of the usual one $\fespex$.  In the next plots, the results for the usual (un-aggregated) \ac{fe} space are labeled as \emph{standard}, whereas the cases with the aggregation are labeled as \emph{aggregated} (or \emph{aggr.} in its short form).
%In the computations with cell aggregation, we take coefficient $\beta$ inversely proportional to the element size with value $\beta=100/h$, and we do not compute any cell-wise eigenvalue problem. For the cases without cell aggregation we compute  cell wise coefficients $\beta_\cell$ solving the generalized eigenvalue problems \eqref{eq:coer-cond}. 
In all the examples, we use Lagrangian reference \acp{fe} with  bi-linear and bi-quadratic shape functions in 2D, and tri-linear and tri-quadratic ones in 3D.

Both the \emph{standard} and the \emph{aggregated} formulations have been implemented in the object-oriented HPC code FEMPAR \cite{badia_fempar:_2017}. The system of 
linear equations resulting from the problem discretization are solved within FEMPAR with a sparse direct solver from the MKL PARDISO 
package \cite{_intel_????}. Condition number estimates are computed outside FEMPAR using the MATLAB function 
{\tt condest}.\footnote{MATLAB is a trademark of THE MATHWORKS INC.} { For the standard unfitted \ac{fe} space $\fespex$, we expect very high condition numbers that can hinder the solution of the discrete system using standard double precision arithmetic. To address this effect and avoid the breakdown of sparse direct solvers, we bound from below the minimum distance between the mesh nodes and the intersection of edges with the boundary $\Gamma$ to a small numerical threshold $D_\mathrm{min}$ proportional to the cell size, namely $D_\mathrm{min}=\varepsilon h$, where $\varepsilon$ is a (mesh independent) user defined tolerance. If the edge cut-node distance is below this threshold , the edge cut is collapsed with the node, perturbing the geometry. In the numerical experiments, we take $\varepsilon=10^{-6}$ and $\varepsilon=10^{-3}$ in 2D and 3D respectively. Using the fact that $\eta_\cell \sim \varepsilon^d$, we can rewrite the condition number estimate \eqref{eq:finite-cell-estimate} in terms of the user-defined tolerance $\varepsilon$ as  $\condnum{\stifmatrix} \sim \varepsilon^{-d(2q +1 -2/d)}$. For instance, we have $\condnum{\stifmatrix} \sim \varepsilon^{-7}$ and $\condnum{\stifmatrix} \sim \varepsilon^{-13}$ for first and second order interpolations, respectively, in 3D. This illustrates that the condition numbers expected for second order interpolation are extremely high as it is confirmed below unless very large values of $\varepsilon$ are considered. However, the value of $\varepsilon$ cannot be increased without affecting the numerical error, since it perturbs the geometry, and destroys at some point the order of convergence of the numerical method. Similar perturbation-based techniques with analogous problems have been used in the frame of the finite cell method in \cite{parvizian_finite_2007}. We note that the tolerance $\varepsilon$ is not needed at all when using the aggregated \ac{fe} space.}

\subsection{Moving domain experiment}
\label{sec:moving-dom}

In the first numerical experiment, we study the robustness of the unfitted \ac{fe} formulations with respect to the relative position between the unfitted boundary 
and the background mesh. To this end, we consider two moving domains that can travel along one of the diagonals of the bounding box
 (see Fig. \ref{fig:mov-test-setup}). The considered geometries are obtained by scaling down the circle and the popcorn flake depicted 
 in Fig. \ref{fig:experimental-setup} by a factor of $0.25$. In both cases, 
the position of the bodies is controlled by the value of the parameter $\ell$ (i.e., the distance between the center of the body and a selected vertex of the box). As the value of $\ell$ varies, the objects move and their relative position with respect to the background 
mesh changes. In this process, arbitrary small cut cells can show up, leading to potential conditioning problems. In this experiment, we consider a background mesh with element size $h=2^{-5}$.
\begin{figure}[ht!]
 \centering
  \begin{subfigure}{0.4\textwidth}
    \includegraphics[width=0.95\textwidth]{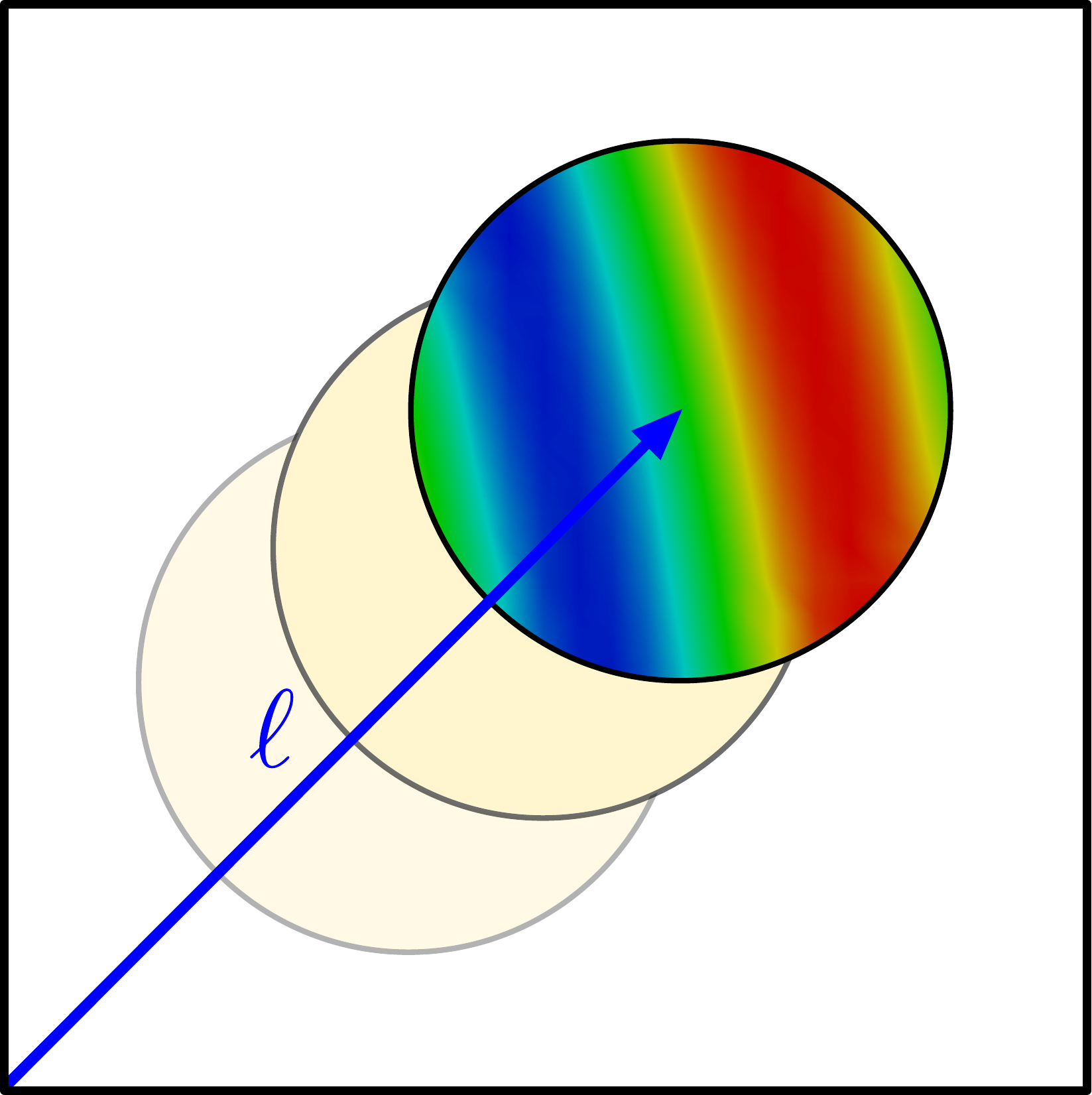}
    \caption{2D case}
    \label{fig:mov-test-setup-a}
  \end{subfigure}
  \begin{subfigure}{0.4\textwidth}
    \includegraphics[width=0.99\textwidth]{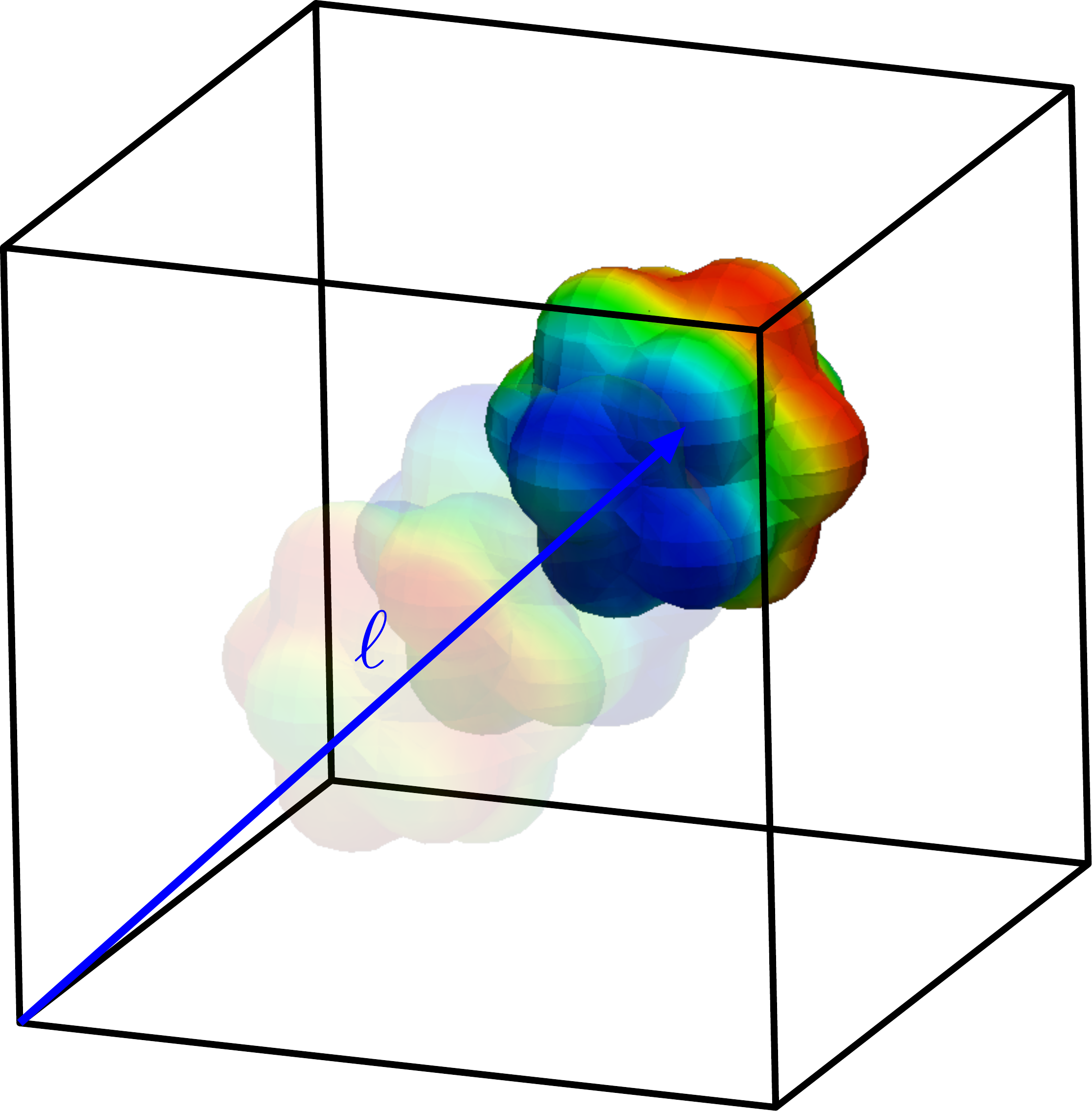}
    \caption{3D case}
    \label{fig:mov-test-setup-b}
  \end{subfigure}
  \begin{subfigure}{0.05\textwidth}
    \includegraphics[width=0.99\textwidth]{fig_colorbar.pdf}
  \end{subfigure}
  \caption{Setup of the moving domain experiment.}
  \label{fig:mov-test-setup}
\end{figure}

Fig. \ref{fig:mov-test-condest} shows the condition number estimate of the underlying linear systems varying the position of the physical domain $\Dom$. The plot is generated using a sample of 200 different values of $\ell$. It is observed  that the condition numbers are very sensitive to the position 
of the domain for the standard unfitted \ac{fe} formulation, whereas the condition numbers are nearly independent of the position when 
using the aggregation-based \ac{fe} spaces. Note that the standard formulation leads to very high condition numbers specially for second 
order interpolations and the 3D case. Moving from 1st order to 2nd order leads to a rise in the condition number between 10 and 
15 orders of magnitude. The same disastrous effect is observed when moving from 2D to 3D.  In contrast, the condition number is nearly insensitive to the number of space dimensions, and mildly 
depends on the interpolation order (as for body-fitted methods) when using aggregation-based \ac{fe} spaces. 
\begin{figure}[ht!]
 \centering
  \begin{subfigure}{0.47\textwidth}
    \includegraphics[width=0.99\textwidth]{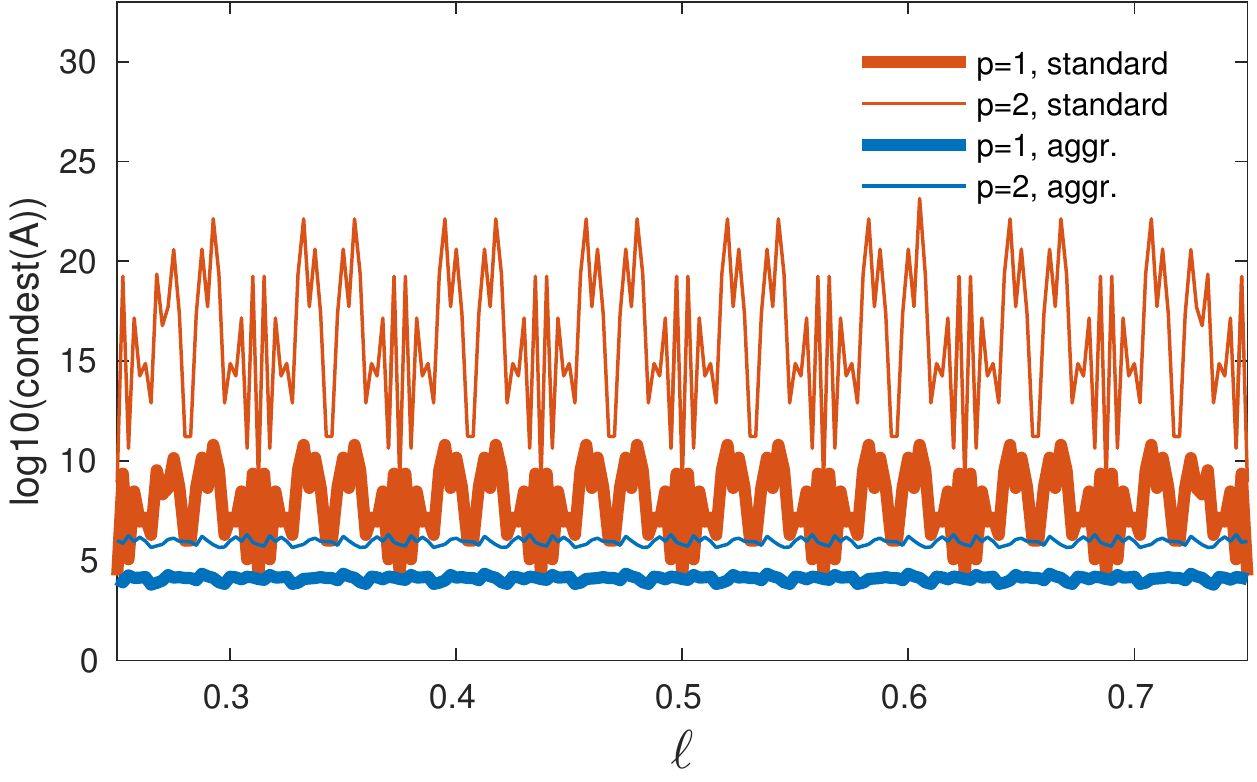}
    \caption{2D case}
    \label{fig:mov-test-condest-a}
  \end{subfigure}
  \begin{subfigure}{0.47\textwidth}
    \includegraphics[width=0.99\textwidth]{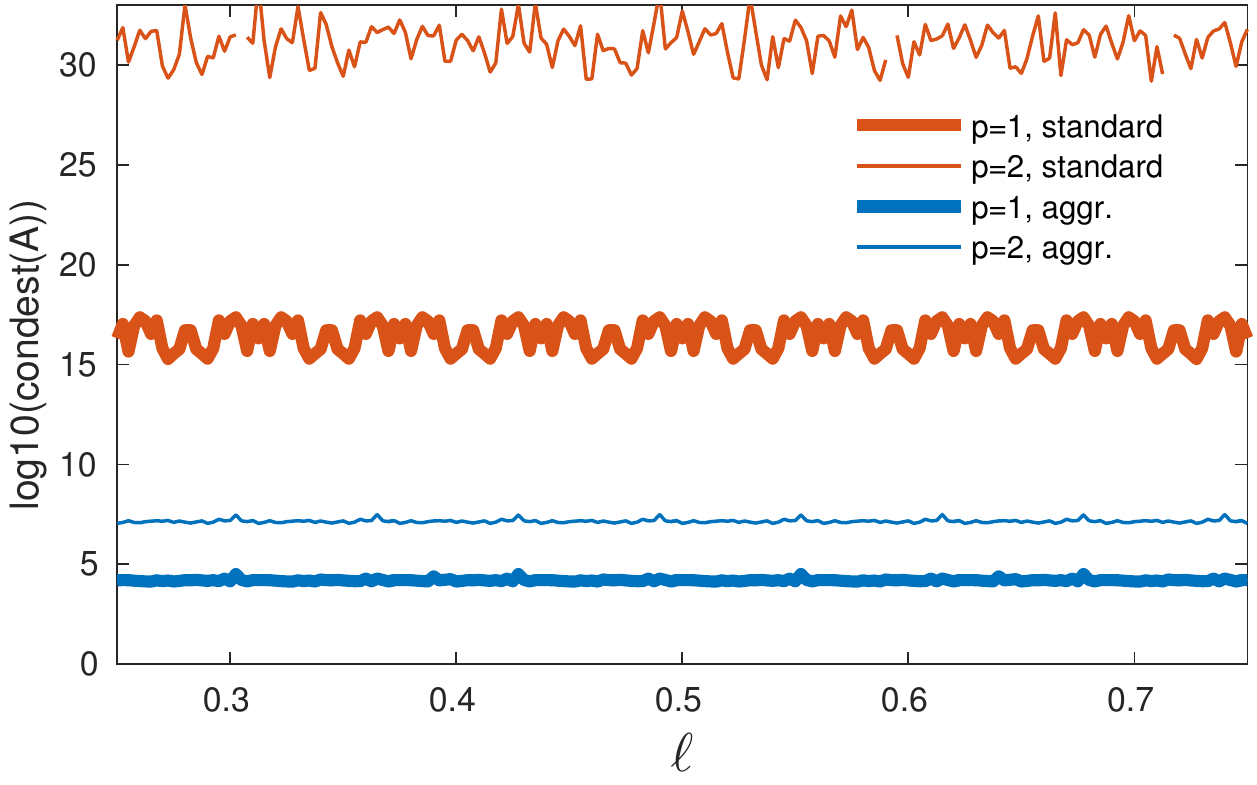}
    \caption{3D case}
    \label{fig:mov-test-condest-b}
  \end{subfigure}
  \caption{Condition number vs. domain position.}
  \label{fig:mov-test-condest}
\end{figure}

From the results shown in Fig. \ref{fig:mov-test-condest}, it is clear that the aggregation-based \ac{fe} spaces are able to dramatically
improve the condition numbers associated with the standard unfitted \ac{fe} formulation. The next question is how cell aggregation impacts on the accuracy of the numerical solution. In order to quantify this effect, Fig. \ref{fig:mov-test-h1err} shows the computed energy norm of the  
discretization error. It is observed that the error is slightly increased when using the aggregation-based \ac{fe} 
spaces. This is because the considered meshes in this moving domain experiment are rather coarse. The error increments become negligible for finer meshes (see Section \ref{sec:conv-test} below).   In this example, we cannot compute a solution for all the values of $\ell$ for  3D  and 
2nd order interpolation without using cell aggregation (see the discontinuous fine red curve in Fig. \ref{fig:mov-test-h1err-b}). The condition numbers are so high (order $10^{30}$) that 
the system is intractable,  even with a sparse direct solver, using standard double precision floating point arithmetic.
\begin{figure}[ht!]
 \centering
  \begin{subfigure}{0.47\textwidth}
    \includegraphics[width=0.99\textwidth]{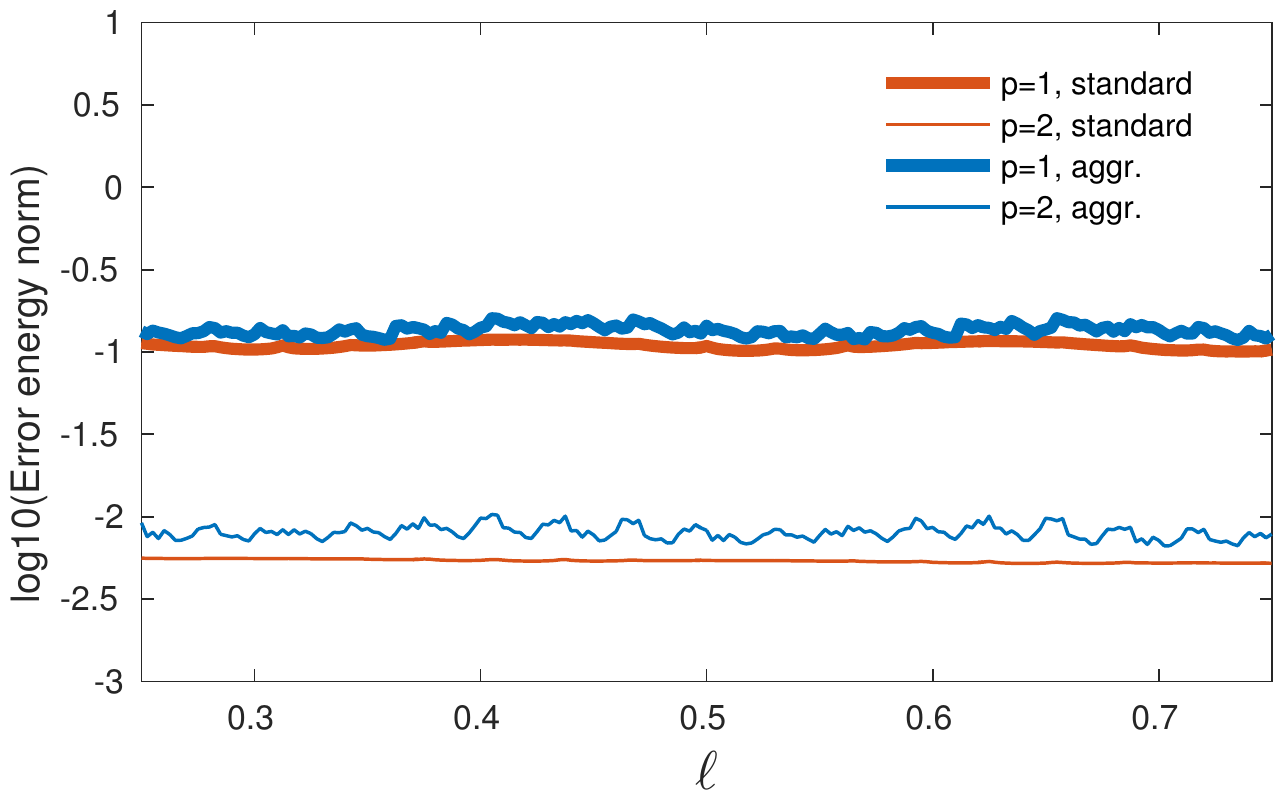}
    \caption{2D case}
    \label{fig:mov-test-h1err-a}
  \end{subfigure}
  \begin{subfigure}{0.47\textwidth}
    \includegraphics[width=0.99\textwidth]{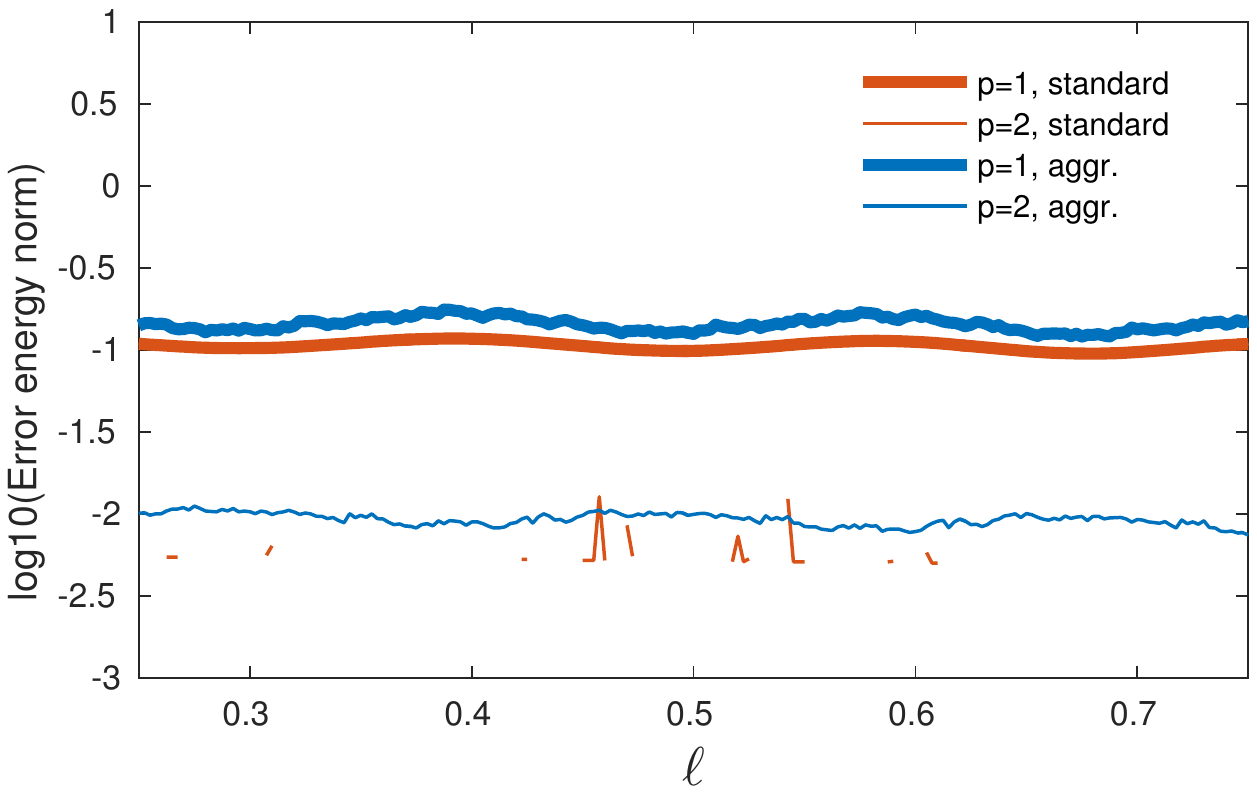}
    \caption{3D case}
    \label{fig:mov-test-h1err-b}
  \end{subfigure}
  \caption{Error energy norm vs. domain position.}
  \label{fig:mov-test-h1err}
\end{figure}

\subsection{Convergence test}
\label{sec:conv-test}
The second experiment is devoted to study the asymptotic behavior of the methods as the mesh is refined. To this end, we consider the geometries and bounding boxes displayed in Fig. \ref{fig:experimental-setup}, which are discretized with uniform Cartesian meshes with element sizes $h=2^{-m}$,  $m=3,4,\ldots,9$ in 2D, and $m=3,4,5,6$ in 3D. 

First, we study how  the size of the aggregates scales when the mesh is refined. Fig. \ref{fig:aggr-conv} shows that the aggregate size is $2h$ in 2D, whereas it tends to $3h$ in the 3D case. These results agrees with the theoretical bounds for the aggregate size discussed in Section~\ref{sec:bac_mesh}.
\begin{figure}[ht!]
  \centering
  \includegraphics[width=0.58\textwidth]{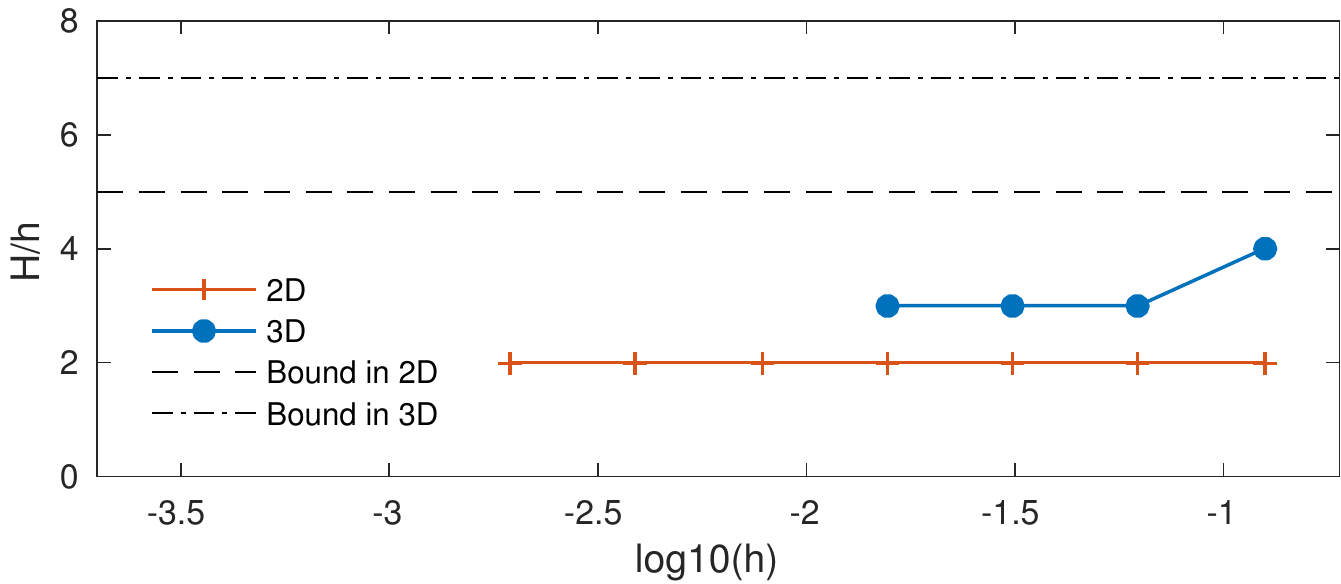} 
  \caption{Scaling of the ratio maximum aggregate size ($H$) / element size ($h$).}
  \label{fig:aggr-conv}
\end{figure}

Then, we study the scaling of the condition numbers with respect to the mesh size (see Fig.~\ref{fig:conv-test-condest}).  For the aggregation-based \ac{fe} spaces, the condition numbers of the stiffness matrix scales as $h^{-2}$, like in standard \ac{fe} methods for body fitted meshes. This confirms the theoretical result of Corollary \ref{cr:cn}. Conversely, the condition number has an erratic behavior if cell aggregation is not considered. The reason is that, as shown in the previous experiment (cf. Section \ref{sec:moving-dom}), the standard unfitted \ac{fe} formulation leads to condition numbers very sensitive to the position of the unfitted boundary. Several configurations of cut cells can show up when the mesh is refined, leading to very different condition numbers. As in the previous experiment, the condition number is very sensitive to the interpolation order and number of space dimensions for the standard unfitted \ac{fe} formulations. This effect is reverted when using cell aggregates.
\begin{figure}[ht!]
 \centering
  \begin{subfigure}{0.4\textwidth}
    \includegraphics[width=0.9\textwidth]{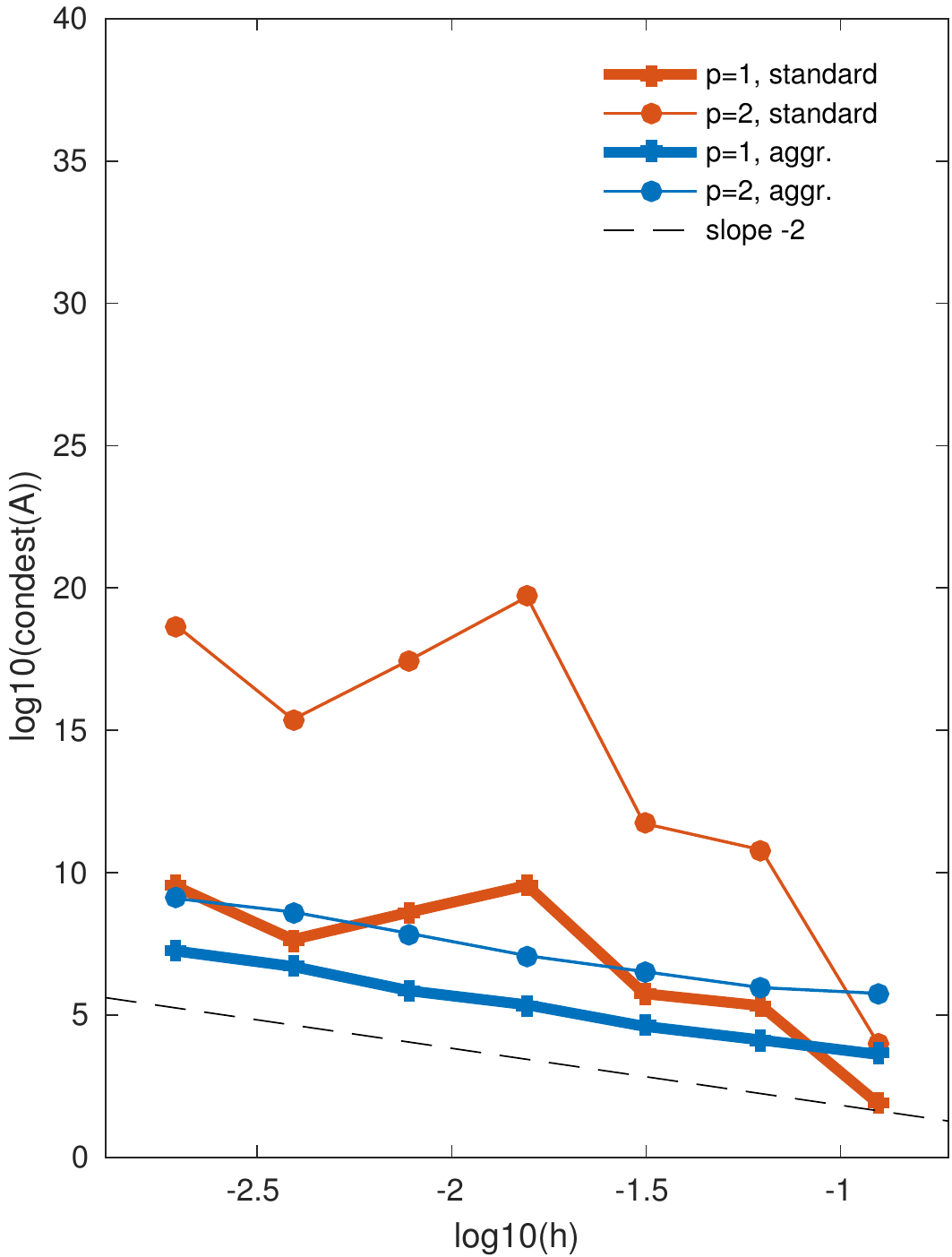}
    \caption{2D case}
    \label{fig:conv-test-condest-a}
  \end{subfigure}
  \begin{subfigure}{0.4\textwidth}
    \includegraphics[width=0.9\textwidth]{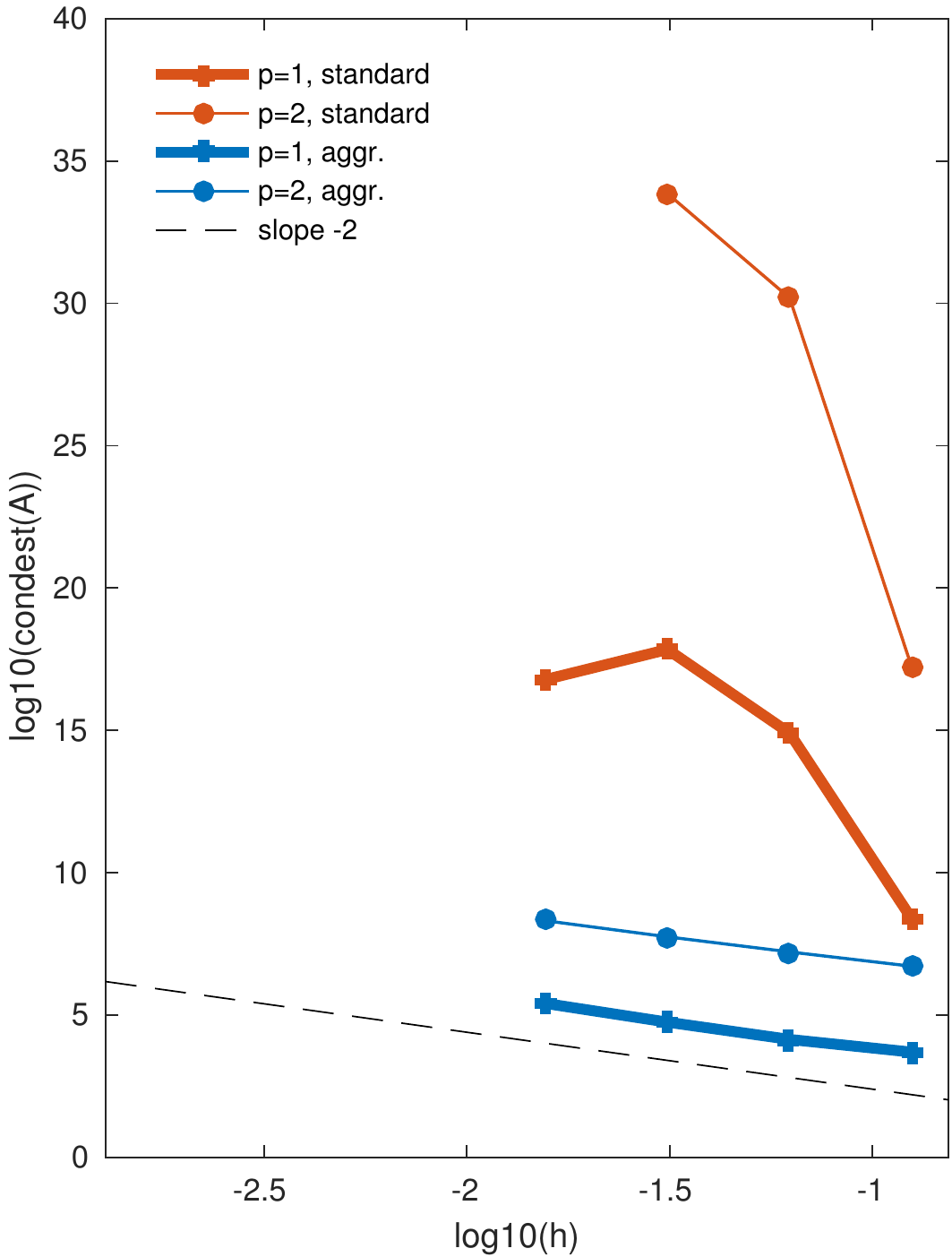}
    \caption{3D case}
    \label{fig:conv-test-condest-b}
  \end{subfigure}
  \caption{Scaling of the condition number upon mesh refinement.}
  \label{fig:conv-test-condest}
\end{figure}

Finally, we study the convergence of the discretization error. To this end, Figs. \ref{fig:conv-test-h1err} and \ref{fig:conv-test-l2err} report the discretization errors measured both in the energy norm and in the $L^2$ norm. Like in the previous experiment (cf. Section~\ref{sec:moving-dom}), the discrete system could not be solved when using the finest meshes in 3D for 2nd order interpolation without using cell-aggregation due to extremely large condition numbers (see the incomplete curve in Fig. \ref{fig:conv-test-h1err-b}).  The results show that the error increment associated with the aggregation-based \ac{fe} space becomes negligible when the mesh is refined. Moreover, the theoretical results of Section \ref{sec:err_est} are confirmed: optimal order of convergence is  (asymptotically) achieved in all cases when using aggregation-based \ac{fe} spaces both for the energy and the $L^2$ norm, for 1st and 2nd order interpolations, and for 2 and 3 spatial dimensions. %In conclusion, the results show that cell-aggregation strategy is able to dramatically improve the conditioning problems of unfitted FEM formulations, without significantly affecting the approximation properties of the finite element interpolation. 
\begin{figure}[ht!]
 \centering
  \begin{subfigure}{0.4\textwidth}
    \includegraphics[width=0.9\textwidth]{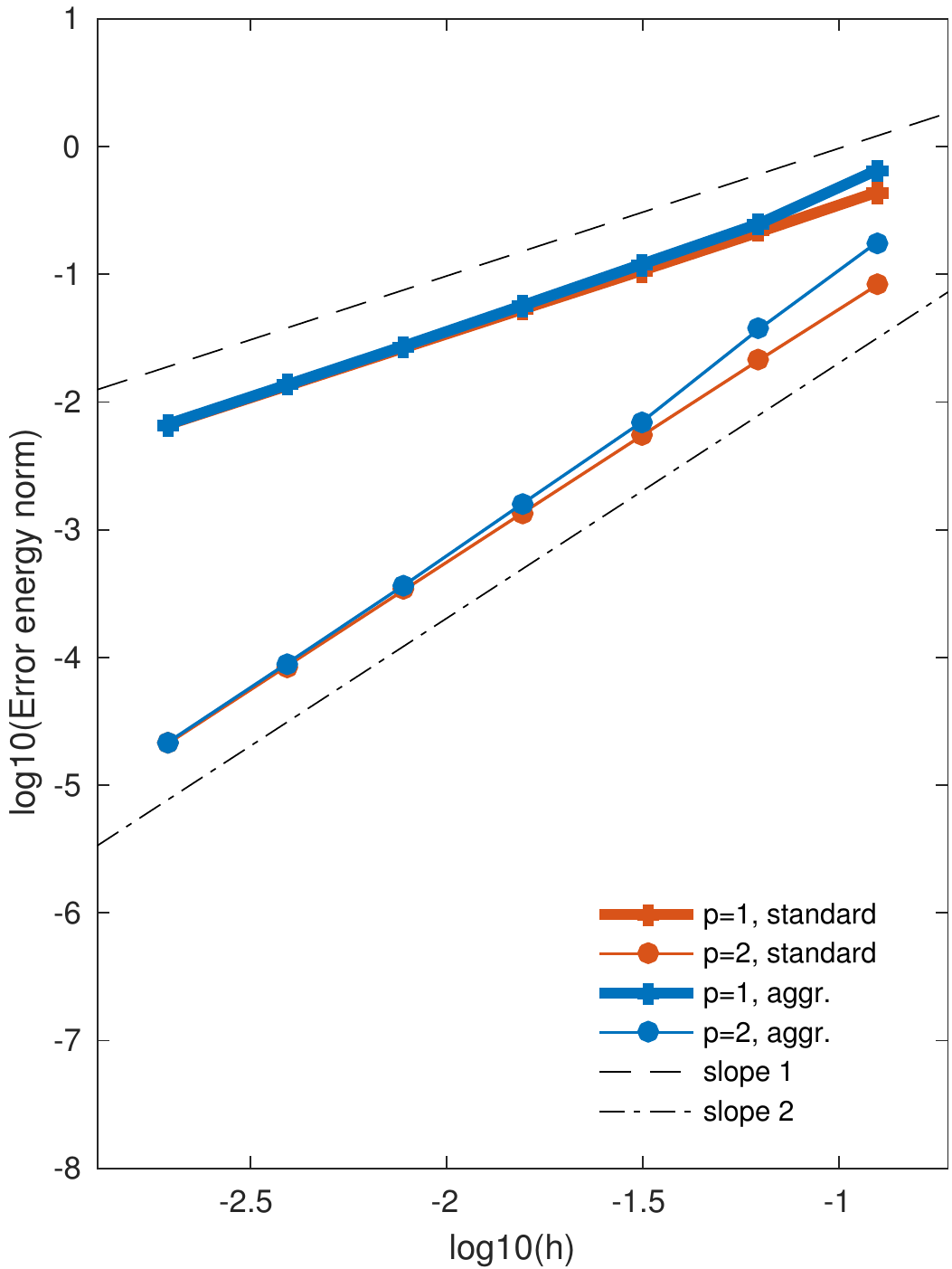}
    \caption{2D case}
    \label{fig:conv-test-h1err-a}
  \end{subfigure}
  \begin{subfigure}{0.4\textwidth}
    \includegraphics[width=0.9\textwidth]{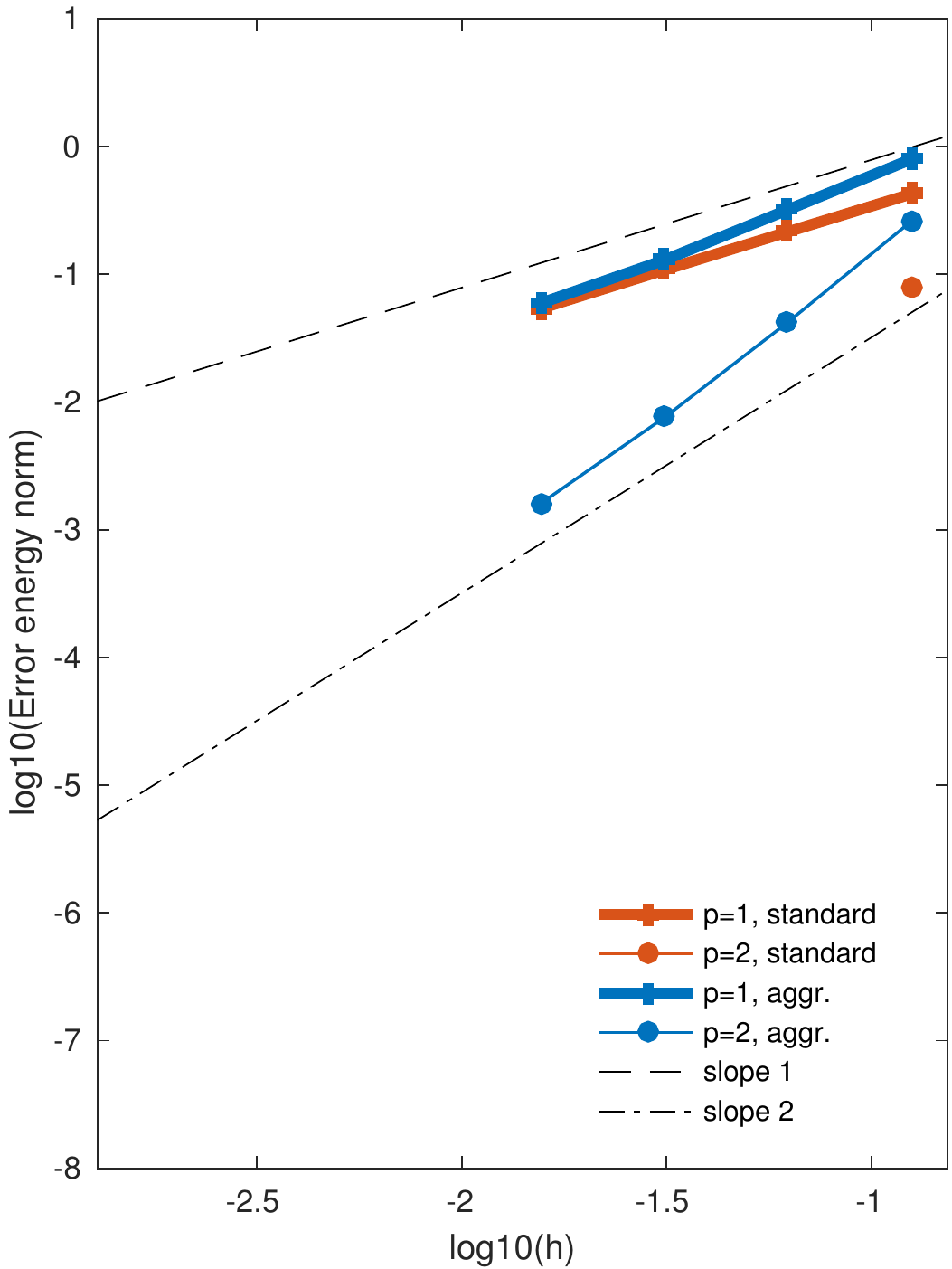}
    \caption{3D case}
    \label{fig:conv-test-h1err-b}
  \end{subfigure}
  \caption{Convergence of the discretization error in energy norm.}
  \label{fig:conv-test-h1err}
\end{figure}

\begin{figure}[ht!]
 \centering
  \begin{subfigure}{0.4\textwidth}
    \includegraphics[width=0.9\textwidth]{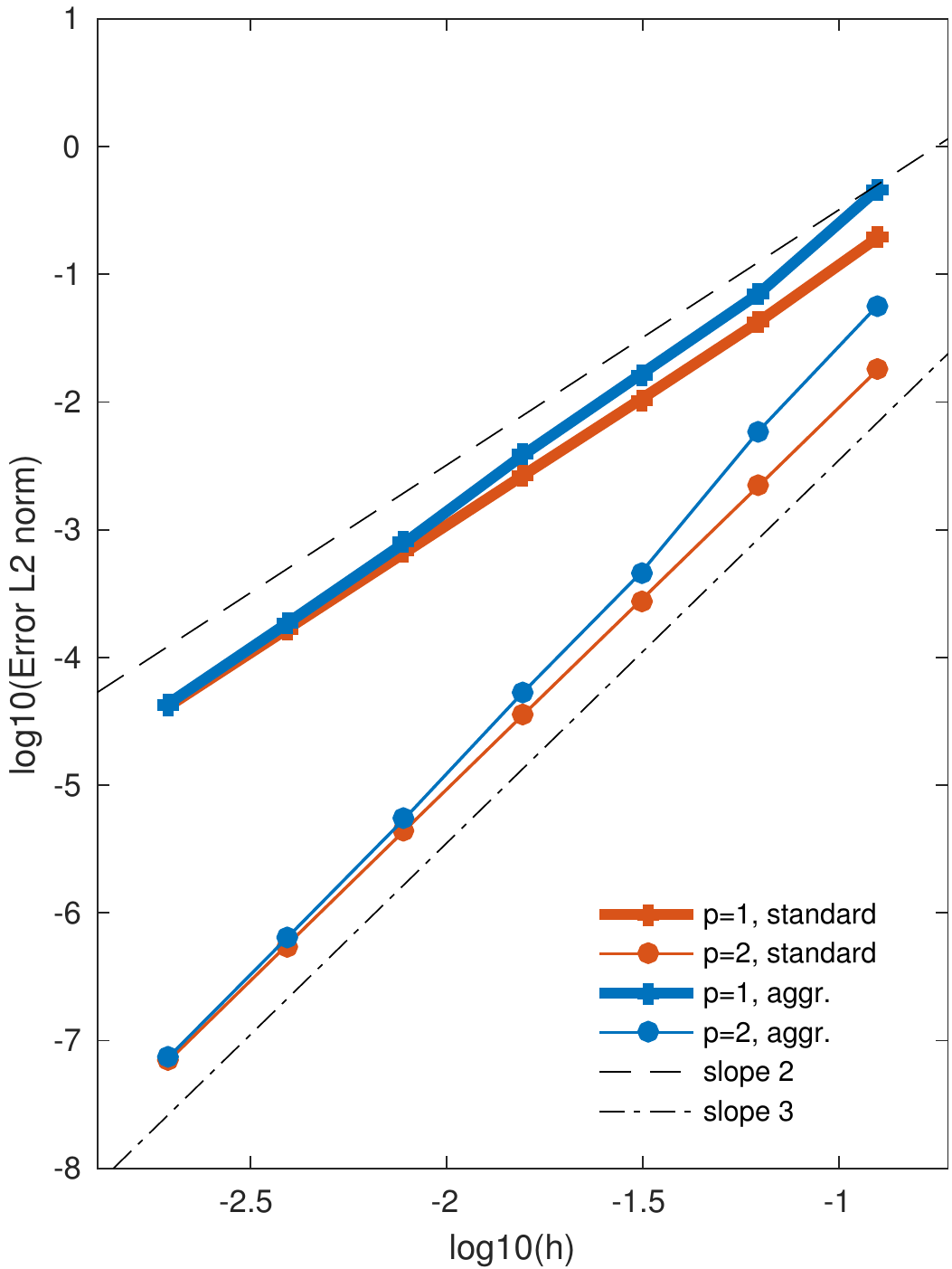}
    \caption{2D case}
    \label{fig:conv-test-l2err-a}
  \end{subfigure}
  \begin{subfigure}{0.4\textwidth}
    \includegraphics[width=0.9\textwidth]{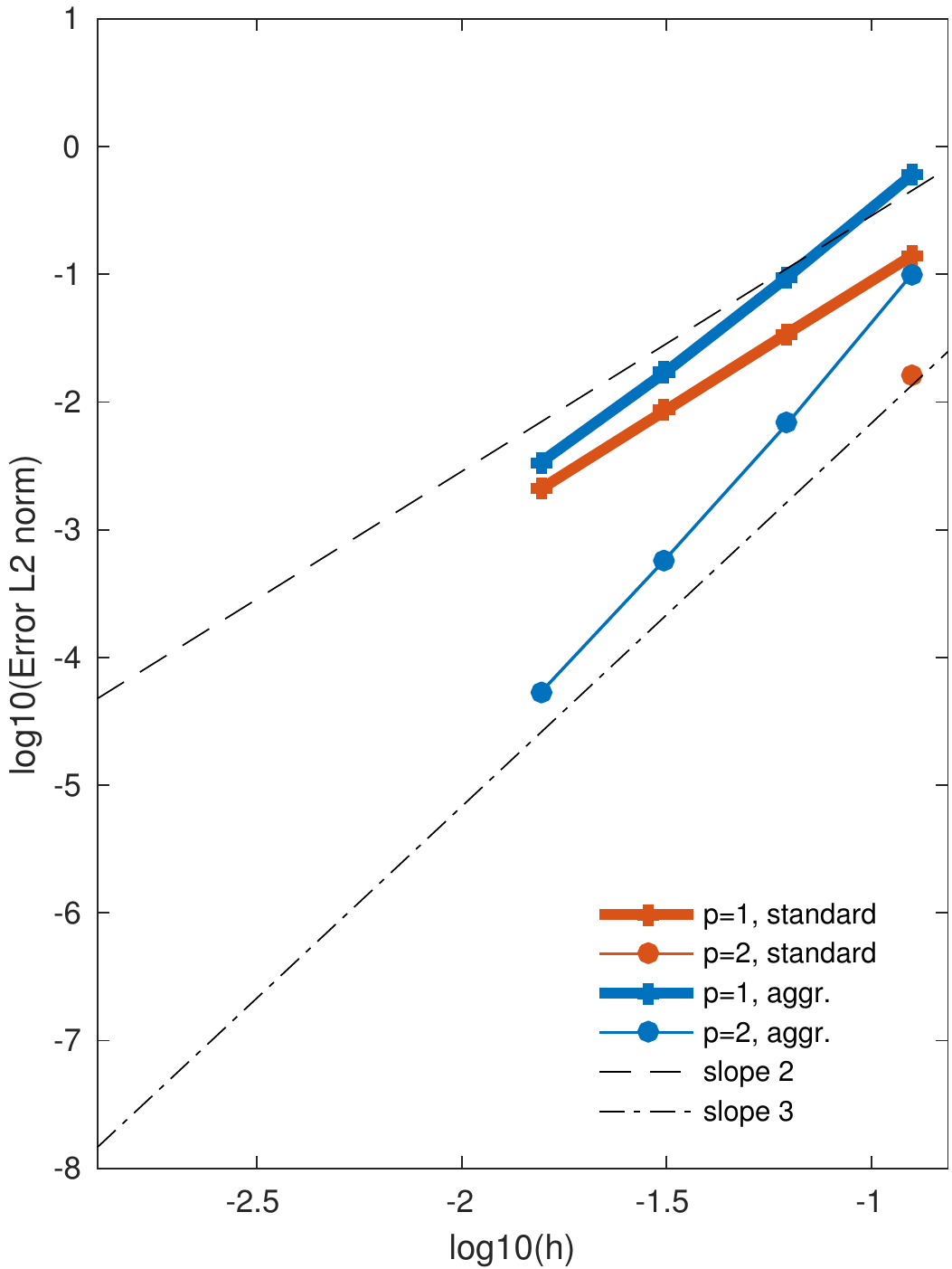}
    \caption{3D case}
    \label{fig:conv-test-l2err-b}
  \end{subfigure}
  \caption{Convergence of the discretization error in $L^2$ norm.}
  \label{fig:conv-test-l2err}
\end{figure}

\section{Conclusions}\label{sec:concl}

We have proposed a novel technique to construct \ac{fe} spaces designed to improve the conditioning problems associated with unfitted \ac{fe} methods. The spaces are defined using cell aggregates obtained by merging the cut cells to interior cells. In contrast to related methods in the literature, the proposed technique is easy to implement in existing \ac{fe} codes (it only involves cell-wise constraints) and it is general enough to deal with both continuous and \ac{dg} formulations. Another novelty with respect to previous works is that we include the
mathematical analysis of the method.  For elliptic problems, we have proved that 1) the novel \ac{fe} space leads to condition numbers that are independent from small cut cells, 2) the condition
number of the resulting system matrix scales with the inverse of the square of the size of the background mesh as in standard \ac{fe} methods, 3) the penalty parameter of
Nitsche's method is bounded from above, and 4) the optimal \ac{fe} convergence order is recovered. These
theoretical results are confirmed with 2D and 3D numerical experiments using both first and second order interpolations.

\bibliographystyle{myabbrvnat}
\bibliography{art026}

\end{document}